\documentclass[11pt, a4paper]{article}

\usepackage[utf8]{inputenc}
\usepackage[colorinlistoftodos]{todonotes}
\usepackage{mathdots}
\usepackage{comment}
\usepackage{amsmath,amssymb,amsfonts,amsthm}
\usepackage{authblk}
\usepackage{fullpage}

\newcommand{\Arnaud}[1]{\todo[size=\small,inline,color=purple!30]{#1 \\ \hfill --- Arnaud}}

\newtheorem{lemma}{Lemma}

\newtheorem{theorem}{Theorem}
\newtheorem{proposition}{Proposition}

\theoremstyle{definition}

\theoremstyle{remark}

\author[1]{Erin Wolf Chambers \thanks{This work was funded in part by the National Science Foundation through grants CCF-1614562, CCF-1907612 and DBI-1759807.}}
\author[2]{Francis Lazarus \thanks{This author is partially supported by the French ANR projects GATO (ANR-16-CE40-0009-01) and MINMAX (ANR-19-CE40-0014) and the LabEx PERSYVAL-Lab (ANR-11-LABX-0025-01) funded by the French program Investissement d’avenir.}}
\author[3]{\\Arnaud de~Mesmay\thanks{This author is partially supported by the French ANR projects ANR-17-CE40-0033 (SoS),  ANR-16-CE40-0009-01 (GATO) and ANR-19-CE40-0014 (MINMAX).}}
\author[4]{Salman Parsa\thanks{This work was funded in part by the National Science Foundation through grants CCF-1614562 as well as funding from the SLU Research Institute.}}

\affil[1,4]{Saint Louis University, Saint Louis, MO, USA}
\affil[2]{G-SCOP, CNRS, UGA, Grenoble, France}
\affil[3]{LIGM, CNRS, Univ. Gustave Eiffel, ESIEE Paris, F-77454 Marne-la-Vall\'ee, France}

\date{}
\begin{document}

\bibliographystyle{plain}

\title{Algorithms for Contractibility of Compressed Curves on 3-Manifold Boundaries}




\maketitle

\begin{abstract}
In this paper we prove that the problem of deciding contractibility of an arbitrary closed curve on the boundary of a 3-manifold is in \textbf{NP}. We emphasize that the manifold and the curve are both inputs to the problem. Moreover, our algorithm also works if the curve is given as a compressed word. Previously, such an algorithm was known for simple (non-compressed) curves, and, in very limited cases, for curves with self-intersections. Furthermore, our algorithm is fixed-parameter tractable in the complexity of the input 3-manifold.

As part of our proof, we obtain new polynomial-time algorithms for compressed curves on surfaces, which we believe are of independent interest. We provide a polynomial-time algorithm which, given an orientable surface and a compressed loop on the surface, computes a canonical form for the loop as a compressed word. In particular, contractibility of compressed curves on surfaces can be decided in polynomial time; prior published work considered only constant genus surfaces.
More generally, we solve the following normal subgroup membership problem in polynomial time: given an arbitrary orientable surface, a compressed closed curve $\gamma$, and a collection of disjoint normal curves $\Delta$, there is a polynomial-time algorithm to decide if $\gamma$ lies in the normal subgroup generated by components of $\Delta$ in the fundamental group of the surface after attaching the curves to a basepoint.
\end{abstract}

\section{Introduction}
In a topological space $X$, a closed curve $c:\mathbb{S}^1 \rightarrow X$ is \emph{contractible} if the map $c$ extends to a map of the standard disk, $f:D \rightarrow X$ , $f|_{\partial D}=c$; such a curve can be continuously deformed to a point, and represents the trivial homotopy class for the fundamental group of the space.
Dating back to Dehn~\cite{dehn} more than a century ago, the problem of testing contractibility has been an important focus in the development of efficient algorithms in computational topology and group theory.
Such problems are now well understood in two dimensions, with recent works of the second author and Rivaud~\cite{lazarusrivaud} and Erickson and Whittlesey~\cite{tcs} providing linear-time algorithms to test contractibility and free homotopy of curves on surfaces; Despré and the second author~\cite{desprelazarus} also developed efficient algorithms for the closely related problem of computing the geometric intersection number of curves. On the other hand, in higher dimensions, homotopy problems quickly become undecidable: testing contractibility of curves is already undecidable on $2$-dimensional simplicial complexes and in $4$-manifolds (see Stillwell~\cite[Chapter~8]{stillwell}).

The aim of this paper is to further improve our understanding of the intermediate $3$-dimensional case. In a $3$-manifold, testing whether a curve is contractible is known to be decidable, but the best known algorithms are intricate and rely on the proof of the Geometrization Conjecture. We refer to the survey of Aschenbrenner, Friedl and Wilton~\cite{aschenbrennerfriedlwilton}. When the input manifold is a knot complement, this problem contains as a special case the famous \textsc{Unknot Recognition} problem, which asks whether an input knot in $\mathbb{R}^3$ is trivial. That problem is now known to be in \textbf{NP}~\cite{hasslagariaspippenger} and co-\textbf{NP}~\cite{lackenby}. Recently, Colin de Verdi\`ere and the fourth author~\cite{cdvparsa} initiated the study of a problem in between those two, which is testing the contractibility of a closed curve $\gamma$ which lies on the boundary of a triangulated $3$-manifold $M$, and they provided an algorithm to solve this problem in time $2^{O(|M|+|\gamma|)^2})$. They asked whether the problem lies in the complexity class \textbf{NP} and proved that it holds in two cases: when the number of self-intersections of $\gamma$ is at most logarithmic, and when the boundary surface is a torus. We remark that the existence of an algorithm to test contractibility of a closed curve on the boundary actually follows from the earlier results of Waldhausen~\cite{waldhausen}, and an exponential-time algorithm can be derived from the standard 3-manifold literature, as we will explain later in this paper. However, the algorithm of~\cite{cdvparsa,cdvparsaJ} is simpler in the sense that it relies only on the proof of the Loop Theorem.

Our main result is that the problem of contractibility when the curve is on the boundary of the 3-manifold is in \textbf{NP}. Furthermore, the algorithm that we provide has two additional features. First, it is fixed parameter tractable in the input manifold: for a given $3$-manifold $M$, after a preprocessing taking exponential time in $|M|$, we can solve any contractibility test for curves on the boundary in polynomial time. Second, our algorithm also works if the input curve is \textit{compressed} via a straight-line program, in particular, it can be exponentially longer than its encoding size (we defer to Section~\ref{S:background} for the precise definitions).

\begin{theorem}\label{t:3mcontract}
Let a triangulated $3$-manifold $M$ and a curve $\gamma$ on the boundary of $M$ be given as the input. The problem of deciding whether $\gamma$ is contractible in $M$ is in \textbf{NP}, and there is a deterministic algorithm that solves the problem in time $2^{O(|M|^3)}poly(|\gamma|)$. This is also the case if $\gamma$ is given as a compressed word, i.e., a straight-line program.
\end{theorem}

Our proof of Theorem~\ref{t:3mcontract} will reduce the problem to another problem involving only curves on surfaces. However, there is an important subtlety. Many algorithms in $3$-dimensional topology involve an exponential blow-up in the size of the objects under consideration. This is the case for example for \textsc{Unknot Recognition}, where the classical algorithm~\cite{hasslagariaspippenger} looks for a disk spanning the knot, but this disk might be exponentially large~\cite{hasssnoeyinkthurston}. In order to get \textbf{NP} membership, this exponential blow-up is controlled using a compressed\footnote{There are two different forms of compression going on in this paper: normal coordinates and straight line programs. We keep the name \textit{compressed} for the latter, while a curve or a multicurve encoded with normal coordinates will be simply called a \textit{normal (multi-)curve}.} system of coordinates called \textit{normal coordinates} (see Section~\ref{S:background} for definitions). Our approach follows a similar scheme and suffers from a similar blow-up, which forces us to deal with normal curves. Precisely, we show that the problem of contractibility of an arbitrary curve on the boundary of a 3-manifold reduces to the following problem. Given a family of disjoint normal closed curves $\Delta$ on a triangulated surface $S$, and a curve $\gamma$ on $S$, decide if $\gamma$ lies in the normal subgroup of the fundamental group determined by the curves of $\Delta$ (appropriately connected to the basepoint). This is equivalent to deciding if the curve $\gamma$ is contractible in the space resulting from $S$ by gluing a disk to each component of $\Delta$. We call this problem the \textit{disjoint normal subgroup membership} problem, and we provide an algorithm to solve it in polynomial time, despite the highly compressed nature of $\Delta$ and $\gamma$.

\begin{theorem}\label{t:nsm}
Let an orientable triangulated surface $S$, a closed normal multi-curve $\Delta$, and a compressed curve $\gamma$ be given as the input. There is a polynomial-time algorithm for deciding the disjoint normal subgroup membership problem for $\gamma$ and $\Delta$ on $S$.
\end{theorem}

To solve the disjoint normal subgroup membership problem, one of our technical tools is a polynomial-time algorithm to test triviality of a compressed closed curve on a surface, which might be of independent interest. In fact, we compute a canonical form for such a curve, in the sense that there is a unique canonical representative in each based homotopy class.

\begin{theorem}\label{t:2mccontract}
Let an orientable triangulated surface $S$ and a compressed loop $\gamma$ on $S$ described by a straight-line program be given. Then one can transform $\gamma$, in polynomial time, into a canonical form in its based homotopy class, given as a compressed word. In particular, we can test in polynomial time whether $\gamma$ is contractible.
\end{theorem}

We emphasize that in Theorem~\ref{t:2mccontract} the surface $S$ is part of the input. The analogous theorem, if we assume that surface $S$ is not part of the input, has been known in a much broader context, as we discuss in the next section.

\subsection{Related work}

\subparagraph*{Algorithms in $3$-manifold topology} There is now a large body of research on the computational complexity of problems in $3$-manifold topology and knot theory. To paint a picture in broad strokes, most topological problems are now known to be decidable, using a combination of geometric, algebraic and topological tools and the complexity of the best known algorithms range from exponential~\cite{hasslagariaspippenger,agolhassthurston} to quite galactic, see for example Kuperberg~\cite{kuperberghomeomorphism} for the homeomorphism problem. In contrast, only very few complexity lower bounds are known~\cite{agolhassthurston,lackenbyhard,unbearable} and for many problems, including ours, it is open whether a polynomial-time algorithm exists. One particular tool that our work relies on is \textit{normal surface theory}, which provides a powerful framework to enumerate and analyze the topologically relevant embedded surfaces in a $3$-manifold. Actually, tools from normal surface theory~\cite{schubert} provide an off-the-shelf algorithm proving that deciding contractibility of a \textit{simple} curve on the boundary of a $3$-manifold is in \textbf{NP}, as explained in~\cite[Section~3]{cdvparsa}. However, normal surface theory is ill-adapted to study surfaces that are not embedded~\cite{immersed}, hence the more technical path that we follow in this paper.  We refer to Lackenby~\cite{lackenbysurvey} for a recent and extensive survey on algorithms in $3$-manifold topology.

\subparagraph*{Geometric and combinatorial group theory} The problem of deciding the contractibility of a curve in a topological space, which is given for example as a finite simplicial complex, can be equivalently rephrased as a problem of testing the triviality of a word in a finitely presented group, the fundamental group of the space. In the case of 2- and 3-manifolds, the geometry or topology of the underlying manifold has a strong impact on the properties of the group, and there is a vast amount of literature on this interaction, see for example~\cite{delaharpe}. We refer to Aschenbrenner, Friedl and Wilton~\cite{aschenbrennerfriedlwilton,aschenbrennerfriedlwilton2} for the case of $3$-manifold groups. Of particular interest for the word problem is the notion of a \textit{Dehn function}, which upper bounds the area of a disk certifying the triviality of a word, and thus controls the complexity of a brute-force algorithm to solve the word problem. For $3$-manifold groups, this function can be exponential~\cite[Theorem~8.1.3]{wordprocessing}, making such an approach at least exponentially worse than ours. When the Dehn function of a group presentation is polynomial (as is the case, for instance, for the fundamental group of a compression body), it can be used to prove that the word problem is in \textbf{NP}. However, this only works for a fixed presentation. In our problem, the group presentation itself is part of the input, and we cannot use the bound on the Dehn function directly.

\subparagraph*{Algorithms for compressed curves and surfaces} As we hinted at in the introduction, $3$-dimensional algorithms naturally tend to involve exponentially large objects, which are generally compressed using normal coordinates. This incurs a need for efficient algorithms dealing with normal curves and surfaces, even for the most basic tasks. For example, it is not easy to test whether a curve described by normal coordinates is connected. By now, there are many known techniques to handle topological problems on compressed curves, see for example Agol, Hass and Thurston~\cite{agolhassthurston}, Schaefer, Sedgwick and Stefankovi\v{c}~\cite{sss1,sss2}, Erickson and Nayyeri~\cite{Erickson2013}, Bell~\cite{bell} or Dynnikov~\cite{dynnikov2} and Dynnikov and Wiest~\cite{dynnikovwiest}. None of these techniques seem to directly solve our disjoint normal subgroup membership problem, but we rely extensively on the work of Erickson and Nayyeri~\cite{Erickson2013} as a subroutine (see Section~\ref{S:retriangulation}). However, normal coordinates do not describe curves with self-intersections, which are our main object of interest in this paper. This is why we rely on a different compression model in the form of \textit{straight-line programs}. The use of such programs in geometric group theory is not new, and in particular there is a sketch of an algorithm to test triviality of compressed curves on a surface of genus $3$ in an appendix of Schleimer~\cite[Appendix~A]{Schleimer2008}. His approach seems to be readily generalizable to higher genus, but the dependency on the genus is not made explicit. More recently, Holt, Lohrey and Schleimer~\cite{holtlohreyschleimer} provided a polynomial-time algorithm to test triviality of compressed words in \textit{hyperbolic groups}, of which (most) surface groups are a subclass. The difference with our Theorem~\ref{t:2mccontract} is that in their result, the group presentation is not part of the input.
While both approaches could plausibly be used in our setting, we rely on different tools to provide a complete proof of Theorem~\ref{t:2mccontract}: our approach treats the surface as an input and works for any genus, and seamlessly generalizes to the case of wedges of surfaces that we also need.

\subsection{Summary of our techniques}
The standard methods of normal surface theory allow us to reduce our main contractibility problem to the case where the input manifold belongs to a particular class of manifolds called \textit{compression bodies}, see Section~\ref{S:3dstuff}. The fundamental group of compression bodies is roughly a free product of surface groups, and thus one can rely on known algorithms for surfaces~\cite{tcs,lazarusrivaud} to solve the contractibility problem there. However, due to the exponential size of normal surfaces, this would only yield an exponential time algorithm. In order to improve on this, we identify the precise problem that we need to solve to be the disjoint normal subgroup membership problem and set out on a quest to solve it efficiently. While an \textbf{NP} algorithm would suffice for our main result, we will develop a polynomial-time algorithm.

The disjoint normal subgroup membership problem is made delicate by the compressed nature of the curves in $\Delta$. If they were given as disjoint curves, say, on the $1$-skeleton of the surface, we could glue a disk on each of them, and observe that the resulting complex is homotopically equivalent to a wedge of surfaces\footnote{A \textit{wedge} of a family topological spaces is the space obtained after attaching them all to a single common point. Actually, there are also circle summands here -- we do not mention them in this outline to keep the discussion light.}, allowing us to reduce the problem to the triviality of a curve in a surface (see Section~\ref{S:cappingoff} for a description of the homotopy equivalence). In order to do so despite the compression, we compute in Section~\ref{S:retriangulation} a \textit{retriangulation} of the surface, so that the curves in $\Delta$ are only polynomially long in the new triangulation. This is done by tracing the normal curve and building the \textit{street complex} which was specifically designed by Erickson and Nayyeri to handle normal curves on surfaces in polynomial time. We then track how the street complex interacts with our input curve $\gamma$. This operation comes at a cost: even if the input curve $\gamma$ was not compressed, it may become exponentially long after the retriangulation. Since it may be not simple, we rely on straight-line programs instead of normal curves to encode it. At this stage, we note that it might as well have been a compressed straight-line program from the start.

Finally, we want to test the triviality of a compressed curve in a wedge of surfaces. The fundamental group of this wedge is a free product of fundamental groups of surfaces, and thus the crux of the problem is to solve it in a single surface. While there are known techniques to test the contractibility of a curve on a surface, based on local simplification rules~\cite{tcs,lazarusrivaud}, the compressed nature of our curves is once again a stumbling block here, as we need to be careful to never apply an exponential number of such simplification rules. The solution sketched by Schleimer~\cite[Appendix~A]{Schleimer2008} to that issue is to introduce a family of \textit{well-tempered paths}, which are carefully designed to not necessitate such exponential simplifications. We use a different approach: we rely on the standard tools developed for the non-compressed case, namely quad systems and turn sequences, and detect exponential simplification and realize them all at once in polynomial time. This relies on recent insights on the structure of these quad systems~\cite{desprelazarus}~\cite[Section~4.3]{dagstuhl}.

\section{Background}\label{S:background}

We begin by recalling some key definitions and results, although we assume that the reader is familiar with basic algebraic topology such as homotopy and fundamental groups; we refer to Hatcher~\cite{hatcher} or Stillwell~\cite{stillwell} for more detailed definitions.

\subsection{3-dimensional Manifolds}

A $3$-manifold is a topological space where each point is locally homeomorphic to a $3$-dimensional ball or a $3$-dimensional half-ball. The points of the latter type form the \textit{boundary} of the $3$-manifold, which is always a surface. In this paper, $3$-manifolds are always assumed to be triangulated, and we use common and a less restrictive definition than simplicial complexes: a \textit{triangulation} of a $3$-manifold $M$ is a collection of $n$ abstract tetrahedra,  along with a collection of gluing pairs which identify some of the $4n$ boundary triangles in pairs.
In particular, we allow two faces of the same tetrahedron to be identified. If we add the further restriction that the neighborhood of each vertex is a $3$-ball or a half $3$-ball and no edge gets glued to itself with the opposite orientation, then the resulting space is always a 3-manifold~\cite{moise}.

We briefly recall normal surface theory, which will be used in Section~\ref{S:3dstuff}, and refer the reader to Hass, Lagarias and Pippenger~\cite{hasslagariaspippenger} and Burton~\cite{burton} for a more detailed overview and background of the topic. A \textit{normal isotopy} is an isotopy of a triangulated $3$-manifold that leaves each cell of the triangulation invariant.
A \textit{normal surface} in a triangulation is a properly embedded surface transverse to the triangulation; such a surface decomposes into a disjoint collection of \textit{normal disks}, each embedded in a single tetrahedron. Each disk inside a tetrahedra is either a triangle, which separates one vertex in a tetrahedron from the others, or a quadrilateral which separates two vertices from the other two. In each tetrahedron, there are (up to normal isotopy) four types of such triangles, and three types of quadrilaterals. The normal surface can be reconstructed using only the number of these triangles and quadrilaterals in each tetrahedron, as long the collection of disks satisfies a family of equations  called \textit{matching equations} and has at most one quadrilateral type in each tetrahedron. These numbers are called the \textit{normal coordinates}, and since $n$ bits encode numbers up to $2^n$, these normal coordinates naturally have exponential compression: normal coordinates of polynomial bit-size encode surfaces containing possibly exponentially many normal disks.
Nevertheless, one can compute the Euler characteristic of a normal surface in polynomial time in the bit-size encoding, since it is a linear form~\cite[Formula 22]{hasslagariaspippenger}.
Finally,  \textit{crushing} a triangulation $T$ along a normal surface roughly consists, for each tetrahedron containing at least one quadrilateral, in removing that quadrilateral and gluing its neighbours pairwise along the two pair of edges not intersected by the quadrilateral, as in Figure~\ref{F:quads}. Then one tears apart in multiple components the triangulation at the places where it is not a $3$-manifold. This can be done in polynomial time, and surprisingly, this very destructive process has a well-controlled topological meaning when the normal surface is a sphere~\cite{burton}.

\subsection{Curve and surface representations}
We shall often represent a surface by a graph cellularly embedded in that surface. This embedding is encoded as a \textit{combinatorial} surface thanks to a rotation system over the edges of the graph~\cite{mt-gs-01}. Equivalently, one can define a combinatorial surface as a gluing of polygons whose sides are pairwise identified. When all the polygons are triangles, we speak of a \textit{triangulation}, or a \textit{triangulated surface}. Note that we allow a triangle to have two sides identified after the gluing or two triangles to share two vertices but no edge, etc.

The \textit{complexity} of a surface is  the number of cells (vertices, edges and faces) of the complex induced by the graph embedding. A curve on such a surface can be represented or encoded in a variety of ways. For instance, any curve is homotopic to a curve defined by a walk on the 1-skeleton of the complex. The complexity of the curve is the number of edges in the walk counted with repetition. One can also consider curves in general position avoiding the vertices of the graph and cutting its edges transversely as in the next section, in which case we say  the complexity of a curve is its number of intersection points with the graph.

We next outline two different notions of compressed curve representations.
To repeat the footnote of warning in the introduction: normal coordinates and straight line programs are two different methods for compressed representations, both of which we use in this paper. We keep the name \textit{compressed} for straight line programs, while a curve or a multicurve encoded with normal coordinates will be simply called a \textit{normal (multi-)curve}.

\subparagraph*{Normal curves}
Let $S$ be a triangulated surface. A \textit{multi-curve} is a disjoint family of curves embedded on a $S$.
A multi-curve $c$ is \textit{normal} if it is in general position with respect to the triangulation of $S$ and if every maximal arc of $c$ in a triangle has its endpoints on distinct sides. In particular, no component of a normal curve is contained in a single triangle.
A normal curve may have three types of arcs in a triangle, one for each pair of sides. Counting the number of arcs of $c$ of each type in each of the $t$ triangles of the triangulation, we get $3t$ numbers called the \textit{normal coordinates} of $c$. As for triangulated 3-manifolds, a  \textit{normal isotopy} of $S$ is an isotopy that leaves each cell of the triangulation invariant.
Up to a normal isotopy, $c$ is completely determined by its normal coordinates.
It is thus possible to encode a multi-curve with exponentially many arcs by storing its normal coordinates, which then have polynomial bit-complexity.
A normal multi-curve is \textit{reduced} if no two of its components are normally isotopic.

\subparagraph*{Straight-line programs}
Another method of compressing curves is to think of a curve as a word over an alphabet. In this encoding, a letter of the alphabet corresponds to a directed edge on the surface and juxtaposition of letters corresponds to concatenation of paths. Therefore, we can consider curves as abstract words and to represent them using compressed representations of abstract words. The length of a word $w$, denoted by $|w|$, is the number of letters in it.

A \textit{straight-line program} is a four-tuple $\mathbb{A}=\langle \mathcal{L},\mathcal{A},A_n,\mathcal{P}\rangle$ where:
\begin{itemize}
\item $\mathcal{L}$ is a finite alphabet of \textit{terminal} characters,
\item $\mathcal{A}$ is a disjoint alphabet of \textit{nonterminal} characters,
\item $A_n\in \mathcal{A}$ is the \textit{root}, and
\item $\mathcal{P}=\{A_i \rightarrow W_i\}$ is a sequence of \textit{production rules}, where $W_i$ is a word in $(\mathcal{L} \cup \mathcal{A})^*$ containing only non-terminals $A_j$ for $j<i$.
\end{itemize}

A straight-line program is in \textit{Chomsky normal form} if every production rule either has two non-terminals or a single terminal on the right. Every straight-line program can be put in polynomial time into Chomsky normal form by adding intermediate non-terminals, and thus we will always assume that a straight-line program is in Chomsky normal form. We denote by $w(A)$ the word encoded by a non-terminal character $A$, or sometimes simply by $A$ when there is no confusion. Its length is denoted by $|A|$. We will always use a terminal alphabet $\mathcal{L}$ equipped with an involution $a \mapsto \bar{a}$. The \textit{reversal} of a word $w$ is the word obtained by changing its letter $a$ by its reverse $\bar{a}$ and reversing the order of the letters.

\textit{Composition systems} are straight-line programs of a more advanced type, where productions of the form $P\rightarrow A[i:j]B[k:l]$ are allowed, and the meaning is that $A[i:j]$ represents the subword between the $i+1$th and $j$th letter (both included, starting the numbering at $1$) of the word that $A$ encodes. We take the convention that negative indices count from the end of the word, and that $A[i:]$, $A[i]$ and $A[:j]$ are shorthand respectively for $A[i:-1]$, $A[i:i+1]$ and $A[0:j]$.
While composition systems might seem more powerful than straight-line programs, a theorem of Hagenah~\cite{hagenah} says that given any composition system, one can compute in polynomial time an equivalent straight-line program. Henceforth, we will slightly abuse language and freely use the $[\cdot:\cdot]$ construct in our straight-line programs.

The following theorem summarizes the algorithms on straight-line programs that we will rely on, see Schleimer~\cite{Schleimer2008} or Lohrey~\cite{lohrey2}.

\begin{lemma}\label{algo:main}
Let $\mathbb{A}$ and $\mathbb{X}$ be two straight-line programs, $i$ be an integer and $e$ be a letter in the terminal alphabet of $\mathbb{A}$. Then one can, in polynomial-time,
\begin{itemize}
\item compute the length of $\mathbb{A}$,
\item output the letter $\mathbb{A}[i]$,
\item find the greatest $j$ so that $\mathbb{A}[j]=e$,
\item compute a straight-line program $\overline{\mathbb{A}}$ for the reversal word $\overline{w(\mathbb{A})}$,
\item decide whether $w(\mathbb{A})=w(\mathbb{X})$,
\item find the biggest $k$ such that $\mathbb{A}[:k]=\mathbb{X}[:k]$.
\end{itemize}
\end{lemma}

For a given $2$-complex $K$ (in particular if $K$ is a combinatorial surface), a straight-line program with terminal alphabet $\mathcal{L}$ the set of directed edges of $K$ can encode any closed curve on $K$. We call such an encoding a \textit{compressed curve} or \textit{walk}. Simple operations on $K$ can seamlessly be done while updating a compressed curve appropriately to obtain a homotopic compressed curve in the modified complex. For example, contracting an edge $e$ of $K$ simply boils down to adding a production rule $e \rightarrow \varepsilon$, where $\varepsilon$ is the empty word.
Likewise, deleting an edge $e$ bounding a face $\bar{e}p$, where $p$ is the complementary path in that face, amounts to replacing each occurrence of $e$ by a $p$ via a new production rule $e \rightarrow p$.

When no encoding is specified for a closed (multi-)curve, then it is simply encoded as a (multi-)walk on the one skeleton of the underlying surface or complex.

\section{From contractibility to normal subgroup membership problem}\label{S:3dstuff}

The following proposition reduces the contractibility problem for closed curves on the boundary of an orientable 3-manifold $M$ to the normal subgroup membership problem for a boundary surface of $M$, where the multi-curve $\Delta$ is a reduced normal curve with polynomial complexity. We note that this reduction is the only place where the algorithm is non-deterministic, the rest of our algorithms run in deterministic polynomial time.

\begin{proposition}\label{P:from3dto2d}
Let $M$ be an orientable triangulated 3-manifold with boundary. There exists a reduced normal multi-curve $\Delta$ on $\partial M$ with normal coordinates of linear bit-size (with respect to $|M|$), so that for any curve $\gamma$ on $\partial M$, $\gamma$ is contractible in $M$ if and only if $\gamma$ is contained in the normal subgroup $N$ of $\pi_1(\partial M)$ generated by the homotopy classes of the components of $\Delta$.

Further, given $M$, $\Delta$, and a cubic-sized certificate (in $|M|$), there is a polynomial-time algorithm to verify that $\Delta$ has the property that all curves in $N$ are contractible in $M$.
\end{proposition}

\begin{proof}
A $3$-manifold $M$ is \textit{irreducible} if every properly embedded $2$-sphere in $M$ bounds a $3$-ball. Our first step is to transform $M$ so that it is irreducible. A \textit{reducing sphere} is a properly embedded $2$-sphere that does not bound a $3$-ball. If there exists a reducing sphere, there exists one that is normal, with normal coordinates of linear bit-size~\cite[Proposition~3.3.24 and Theorem~4.1.12]{matveev}. One can crush~\cite{burton,jacorubinstein} along such a normal surface to obtain a new triangulated manifold $M'$ with strictly less tetrahedra. Note that since our $3$-manifold is orientable, there are no $2$-sided projective planes, and thus the crushing is well-defined.

We claim that $M'$ has boundary identical to that of $M$. Indeed, a normal sphere does not have quadrilateral faces in the tetrahedra with at least one face on the boundary, and thus those are unharmed by the crushing procedure. Furthermore, the tetrahedra with one edge on the boundary can only carry a single quadrilateral type, whose crushing along glues together the adjacent tetrahedra with faces on the boundary, see Figure~\ref{F:quads}. Since the boundary stays the same, crushing only undoes connected sums and removes connected summands~\cite[Corollary~5]{burton}.
Since connected sums induce free products of fundamental groups, a curve on $\partial M'$ is contractible in $M'$ if and only if it is contractible in $M$. We iterate this procedure on $M'$ until we get a $3$-manifold $M^i$ that is irreducible.

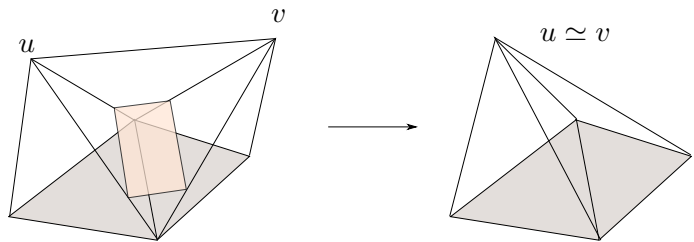
\begin{figure}
    \centering
    \def\svgwidth{9cm}
    \begingroup%
  \makeatletter%
  \providecommand\color[2][]{%
    \errmessage{(Inkscape) Color is used for the text in Inkscape, but the package 'color.sty' is not loaded}%
    \renewcommand\color[2][]{}%
  }%
  \providecommand\transparent[1]{%
    \errmessage{(Inkscape) Transparency is used (non-zero) for the text in Inkscape, but the package 'transparent.sty' is not loaded}%
    \renewcommand\transparent[1]{}%
  }%
  \providecommand\rotatebox[2]{#2}%
  \ifx\svgwidth\undefined%
    \setlength{\unitlength}{3292.43978944bp}%
    \ifx\svgscale\undefined%
      \relax%
    \else%
      \setlength{\unitlength}{\unitlength * \real{\svgscale}}%
    \fi%
  \else%
    \setlength{\unitlength}{\svgwidth}%
  \fi%
  \global\let\svgwidth\undefined%
  \global\let\svgscale\undefined%
  \makeatother%
  \begin{picture}(1,0.32854434)%
    \put(0,0){\includegraphics[width=\unitlength,page=1]{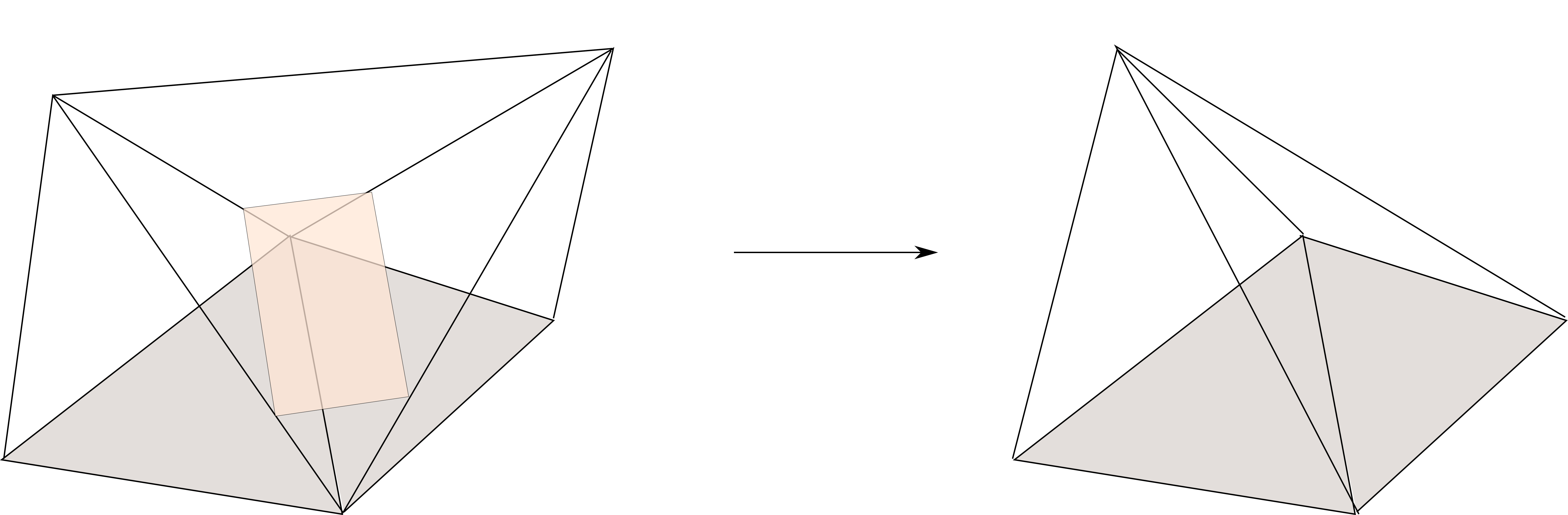}}%
    \put(0.01557014,0.2784227){\color[rgb]{0,0,0}\makebox(0,0)[lb]{\smash{$u$}}}%
    \put(0.3847236,0.32162152){\color[rgb]{0,0,0}\makebox(0,0)[lb]{\smash{$v$}}}%
    \put(0.77744003,0.29413137){\color[rgb]{0,0,0}\makebox(0,0)[lb]{\smash{$u \simeq v$}}}%
  \end{picture}%
\endgroup%

    \caption{The boundary is shaded in gray. Tetrahedra with a triangle on the boundary do not contain quadrilateral normal disks, but tetrahedra with one edge on the boundary might contain one type of those. Crushing along this quadrilateral glues the adjacent tetrahedra together and thus preserves the boundary.}
    \label{F:quads}
\end{figure}

A \textit{compression disk} for the boundary of a $3$-manifold $M$ is a topological disk $D$ that is properly embedded and such that $\partial D$ is an essential curve in $\partial M$. Classic theorems of normal surfaces (see for example Jaco and Tollfeson~\cite[Section~6]{jacotollefson}) show that there exists a complete family $\mathcal{D}$ of linearly many disjoint compression disks for $M^i$ which are \textit{vertex normal surfaces}.
Moreover, we can assume the set of disks $\mathcal{D}$ is \textit{minimal} in the sense that no subset of it is a complete family of compression disks.

Without delving into the details of normal surface theory, we just summarize the salient features of theses vertex normal disks for our purpose: the elements of $\mathcal{D}$ are disks of which the boundaries are normal curves on $\partial M^i$. Each of these boundaries might be exponentially long, but their coordinates, considered as a single multi-curve $\Delta$, fit in a linear number of bits. Furthermore, this family is complete, i.e., if we denote by $C$ the $3$-manifold which is a regular neighborhood of $\partial M \cup \mathcal{D}$, then $\partial C$ is incompressible on both sides: there are no compression disks for $\partial C$ neither in $C$ nor in $M':=M^i \setminus C$.

Let $\gamma$ be a closed curve on $\partial M=\partial M^i$, and let $N$ be the normal subgroup of $\pi_1(\partial M)$ generated by $\Delta$ (after attaching all the curves at a common base point). We claim that $\gamma$ is contractible in $M$ if and only if it belongs to $N$. Let us assume otherwise. Then the loop theorem~\cite[Theorem~4.2]{hempel} implies that there exists a properly embedded disk $D$ in $M^i$ so that $\partial D$ does not belong to $N$. We pick such a properly embedded disk $D$ so that the number of connected components of $D \cap \partial C$ is minimal. This number is non-zero as otherwise $D$ is contained in $C$ and thus $\partial D$ belongs to $N$. Therefore there is an innermost subdisk $D' \subseteq D$ so that $\partial D' \subseteq \partial C$. This disk $D'$ is either contained in $C$ or in $M'$, in both cases this contradicts the completeness of $\mathcal{D}$.

For the verification part, the certificate consists of $\mathcal{D}$, as well as all the reducing normal spheres $S^j$ and the intermediate manifolds $M^j$. Note that since each crushing strictly reduces the size of the triangulation, there are linearly many of those. One can verify in polynomial time that each $S^j$ is a connected sphere (see Agol, Hass and Thurston~\cite[Section~6]{agolhassthurston}), and that crushing along it yields the next manifold: despite the sphere having potentially exponentially many normal disks, the crushing procedure only looks at the existence of quadrilaterals in tetrahedra, which is done in polynomial time. Finally, one checks that each connected component of $\mathcal{D}$ is a disk. By the previous arguments, this ensures the required property. This certificate has cubic size, since there are a linear number of disks and spheres, each of quadratic bit-size.

We finally note that the multi-curve must be reduced. Otherwise, assume there are two components that are isotopic on the boundary. The two normal disks bounding the curves and the image of the isotopy define a 2-sphere which has to bound a ball, thus contradicting the fact that $\Delta$ is minimal.
\end{proof}

Note that we did not strive to optimize the size of the certificate in this proposition, since it does not matter for our application, so further improvements may be possible. In particular, using enumeration techniques for vertex surfaces~\cite[Section~6]{hasslagariaspippenger}, we suspect it can be shown to be quadratic rather than cubic.

Our next proposition shows that we can reduce our problem to the case of orientable manifolds.
\begin{proposition}\label{P:orientability}
The problem of contractibility of a compressed curve on the boundary of an arbitrary 3-manifold reduces, in polynomial time, to the problem of contractibility of a compressed curve on the boundary of an orientable 3-manifold.
\end{proposition}

\begin{proof}
Let the arbitrary manifold be given as a triangulation $T$ and $c$ be any closed curve lying on $\partial M$. We can construct a triangulation of the orientable double cover $\tilde{M}$ of $M$ from $T$ in polynomial time. If $c$ lifts to a closed curve $\tilde{c}$, then $c$ is contractible in $M$ if and only if any lift $\tilde c$ is contractible in $\tilde{M}$. If $c$ does not lift to a closed curve then $c$ is not contractible in $M$.

To finish the proof we need to present a polynomial-time algorithm that given a compressed curve $\gamma$, computes a compressed curve $\gamma'$ in the double covering such that $p(\gamma')=\gamma$ where $p$ is the covering map. We can assume that the compressed word $\gamma$ is given in Chomsky normal form. We lift the sub-paths of $\gamma$ in the two possible ways starting from the leaf production rules. Therefore, when we have a rule of the form $A \rightarrow A_i A_j$ we already have computed two preimages $\tilde{A}_i$ and $\tilde{A}'_i$, and, similarly, $\tilde{A}_j$ and $\tilde{A}'_j$. We compute a preimages of $A$ by looking at the last vertex of $\tilde{A}_i$, if this vertex equals the first vertex of $\tilde{A}_j$ we set $\tilde{A} = \tilde{A}_i\tilde{A}_j$. Otherwise the vertex has to be the same as the first vertex of $\tilde{A}'_j$, we then set $\tilde{A} = \tilde{A}_i\tilde{A}'_j$. We define $\tilde{A'}$ analogously by starting from $\tilde{A}'_i$. At the end, each of the terminals corresponding to the root terminal defines a preimage of the curve $\gamma$ in the double cover.
\end{proof}

Since non-orientable manifolds can be reduced to orientable ones in polynomial time and the boundary of an orientable 3-manifold is orientable, for the rest of the paper we restrict our attention to only considering orientable surfaces.

\subparagraph*{Hardness of the general normal subgroup membership problem}
If the curves $\Delta$ in the normal subgroup membership problem are not disjoint then this problem is much harder, and possibly undecidable, even if they have few edges.

If we can divide $\Delta$ into two subsets, each of which contains only disjoint curves, then it is not hard to see that the problem is equivalent to contractibility in 3-manifolds. Take an arbitrary triangulated 3-manifold $M$ and thicken its 1-skeleton inside $M$. Define $S$ to be the boundary of the thickening of the 1-skeleton and let $\Delta$ be the "meridians" of the edges, and the intersections of the triangles with $S$. Apply a homotopy to the given cycle to lie onto $S$. Then by gluing disks to $\Delta$ we obtain a space which has the same fundamental group as $M$. Therefore, the contractibility of curves in the interior of a 3-manifold reduces to the normal subgroup membership problem with non-disjoint curves. On the other hand, if we can partition $\Delta$ into two subsets $\Delta_1$ and $\Delta_2$, each of which contains only disjoint curves, then the normal subgroup membership problem reduces to contractibility of a curve in the interior of a 3-manifold by
gluing disks to $\Delta_1$ and $\Delta_2$ on opposite sides of $S$, and then thickening.

If the curves $\Delta$ are arbitrary then the problem is undecidable. We can perform the same reduction as above but using a 4-dimensional manifold instead. It is known that the contractibility problem for 4-manifolds is undecidable, since they contain 4-dimensional thickenings of arbitrary 2-complexes~\cite{stillwell}.

\section{Retriangulation}\label{S:retriangulation}
Given a reduced normal multi-curve $\Delta$ on a triangulated 2-manifold, in this section we present a way to compute a new triangulation of polynomial complexity, in which the curves of $\Delta$ lie in the 1-skeleton and are only polynomially long. Techniques exist  (see e.g., Bell~\cite{bell} or Erickson-Nayyeri~\cite{Erickson2013}) to retriangulate in this way.  However, we also must track our compressed word $\gamma$ in this new triangulation without increasing its (compressed) complexity during the retriangulation process. We remark that, as mentioned in the introduction, there are examples where even if we start with a curve $\gamma$ which has polynomially many edges, the image of $\gamma$ in the new triangulation has exponentially many edges.

\begin{proposition}\label{P:retriangulation}
Let a simplicial complex $K$ triangulating a surface $S$, a reduced normal multi-curve $\Delta$ and a compressed closed walk $\gamma$ on the 1-skeleton of $K$ be given as input. The size of the input is the summation of the complexities of the curves $\gamma, \Delta$ and complexity of $K$. There is an algorithm that in polynomial time computes a new simplicial complex $K'$ triangulating the same surface $S$, a multi-curve $\Delta'$, and a walk $\gamma'$ on the 1-skeleton of $K'$ given as a composition system, such that:
\begin{itemize}
    \item there is a homeomorphism $f: |K| \rightarrow |K'|$, such that the images of $\Delta$ under $f$ coincides with $\Delta'$,
    \item $\gamma'$ is homotopic to $f(\gamma)$,
    \item the multi-curve $\Delta'$ is an embedded multi-curve on the 1-skeleton of $K'$.
\end{itemize}
Consequently, the homotopy class of $\gamma'$ belongs to the normal subgroup generated by the components of $\Delta'$ in $|K'|$ if and only if the homotopy class of $\gamma$ belongs to the normal subgroup generated by the components of $\Delta$ in $|K|$.
\end{proposition}

\begin{proof}
A re-triangulation satisfying the conditions on $\Delta$ and $\Delta'$ and disregarding $\gamma$ can be constructed by at least two explicit polynomial-time algorithms in the literature, one by Erickson and Nayyeri~\cite{Erickson2013} and the second by Bell~\cite{bell}. We work with the former, basing our re-triangulation on their \emph{street complex}, and introduce additional structure to track $\gamma$.

First, we briefly recall their complex, and refer the reader to $\cite{Erickson2013}$ for more details.  The street complex consists of a set of \emph{streets} and \emph{junctions}, which ``bundle" common subpaths of the input curves $\Delta$. These are connected via their common boundary edges in the street complex, refer  We model a street as a rectangle whose boundary is a simple cycle consisted of some vertices and edges. A generic street has two edges on the boundary that are segments of edges of $K$, these are denoted by the $\epsilon_i$ in figure \ref{fig:street} a). The boundary minus $\epsilon_1$ and $\epsilon_2$ is divided into two paths that we call the \textit{sides} of the street. In general, our model of a street is mapped to a street from~\cite{Erickson2013} so that some pairs of the edges might identify with each other. Any pair of identified edges has one edge in each side. We draw these edges using a thick line. Figure \ref{fig:street} b) depicts a street model that is mapped to a spiral, shown in c).

The input compressed curve $\gamma$ is a walk on the 1-skeleton of $K$. We compute a straight-line program in Chomsky normal form for $\gamma$ in polynomial time. The production rules with a terminal at the right-hand side now are of the form $A\rightarrow e$ where $e$ is an oriented edge of the 1-skeleton of $K$.
We trace the multi-curve $\Delta$ using the algorithm of \cite{Erickson2013}, but at the same time define a composition system $\mathbb{E}$ encoding the intersections of the edge $e$ with the street complex for $\Delta.$ We update the composition system $\mathbb{E}$ after each step of the tracing algorithm, as described below.

\begin{figure}[ht]
    \centering
    \includegraphics[scale=0.35]{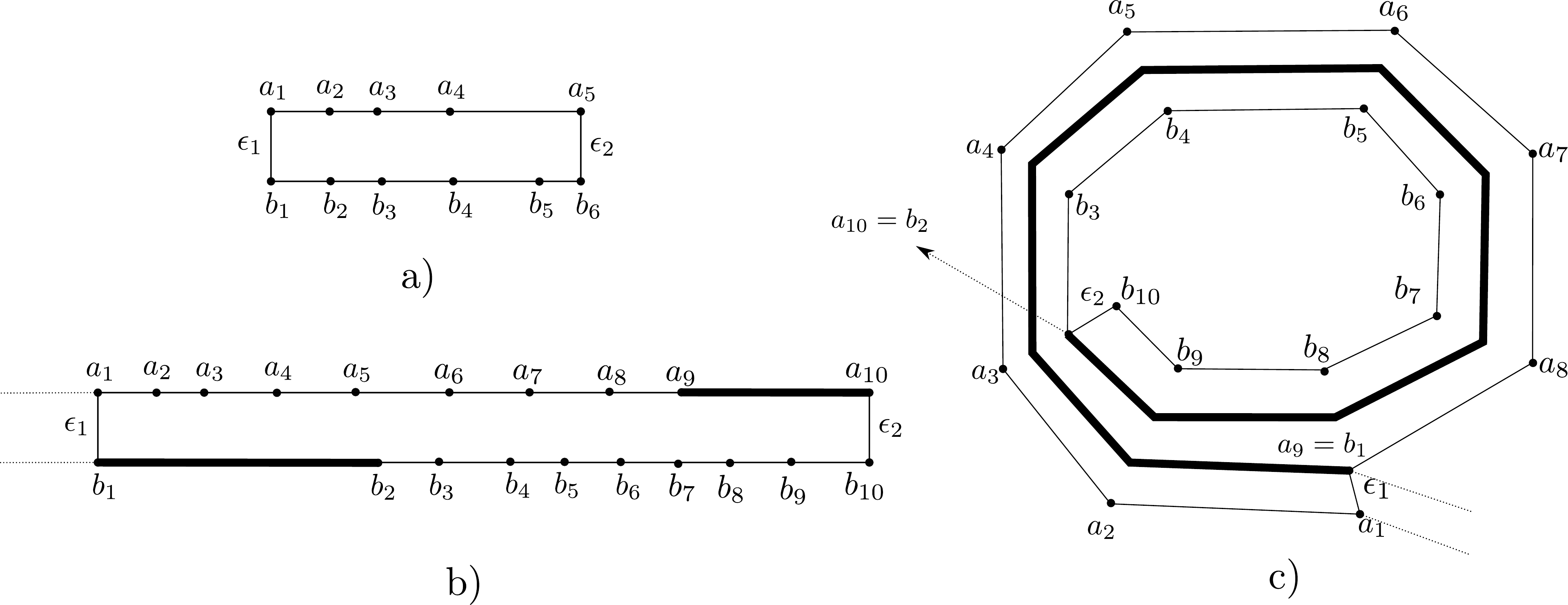}
    \caption{a) A model for a street without a spiral, b) a model for the spiral, where two edges on the boundary (shown in bold) are identified., c) the spiral redrawn as embedded on the surface }
    \label{fig:street}
\end{figure}

The basic \textit{step} of the tracing algorithm is to extend a street across a junction or a fork, and then across a street. The extension across a junction and a street is depicted in Figure~\ref{fig:extension}. This figure depicts a left turn across a junction. When a street is extended through a second street, the latter street is shrunk. When a street is shrunk, one side of the street remains unaltered, we call this side the \textit{fixed side} of the street. The boundary of the street which is extended gains some new vertices and edges of the street complex, while the boundary of the street which is shrunk becomes simpler. A model the effect of this step on the streets is depicted in Figure~\ref{fig:modelextend}.

\begin{figure}[ht]
    \centering
    \includegraphics[scale=0.2]{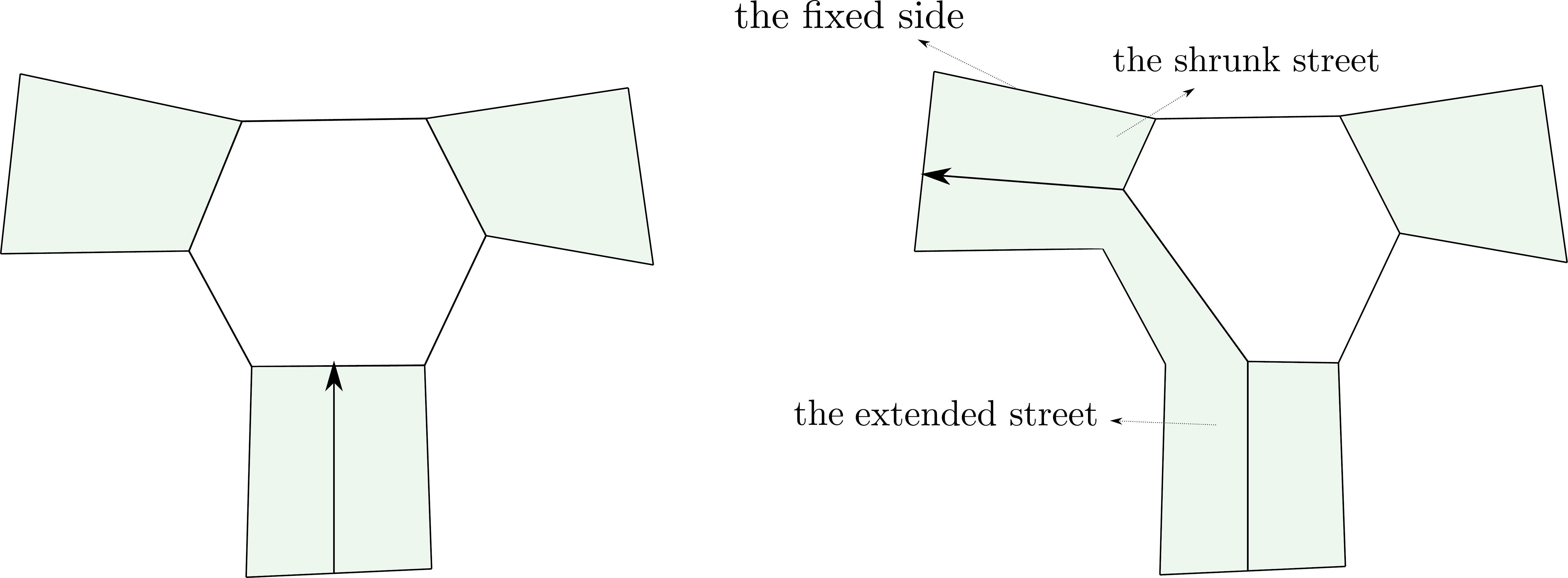}
    \caption{A left turn in the street complex; the arrows depict segments of the edge $e$.}
    \label{fig:extension}
\end{figure}

We next recall that a \textit{phase} of the tracing algorithm of \cite{Erickson2013} consists of a maximal sequence of left turns or of right turns. When we enter a street for the second time, from the same side, during the same phase, we have discovered a \textit{spiral}. In~\cite{Erickson2013}, they show that the spirals can be processed at once in $O(m)$ time, where $m$ is the number of edges of $K$; see Figure~\ref{fig:street}.   The presence of these spirals in a street causes potential identifications; we call the edge which is the result of an identification the \textit{long edge} of the spiral. Observe that when another street extends over a street that contains a spiral, the latter street turns into an ordinary street without identifications. Since we are only interested in the case where $\Delta$ is reduced, the complexity of the street complex is always bounded by a linear function of the complexity of $K$~\cite[Lem. 2.3]{Erickson2013}.

The terminals of $\mathbb{E}$ are of three types: i) $(s,a,x)$ or $(s,x,a)$, where $s$ is a street, $a$ is a vertex on $s$, and $x$ is an edge of the street on the opposite side from $a$, ii) $(s,x,y)$ where $x$ and $y$ are edges on opposite sides on the model of street $s$, and iii) $(s,a,b)$ where $a$ and $b$ are vertices of the model street for $s$ lying on opposite sides. These three types of terminals represent sub-arcs of $e$ that lie inside the street, where the sub-arcs start at the second parameter and end at the third parameter. We use vertices and edges of the model street rather than the street itself. This allows us to distinguish between two arcs in the street complex, which have different directions, but connect the image of two edges or vertices that are identified in the street complex.

\begin{figure}[ht]
    \centering
    \includegraphics[scale=0.4]{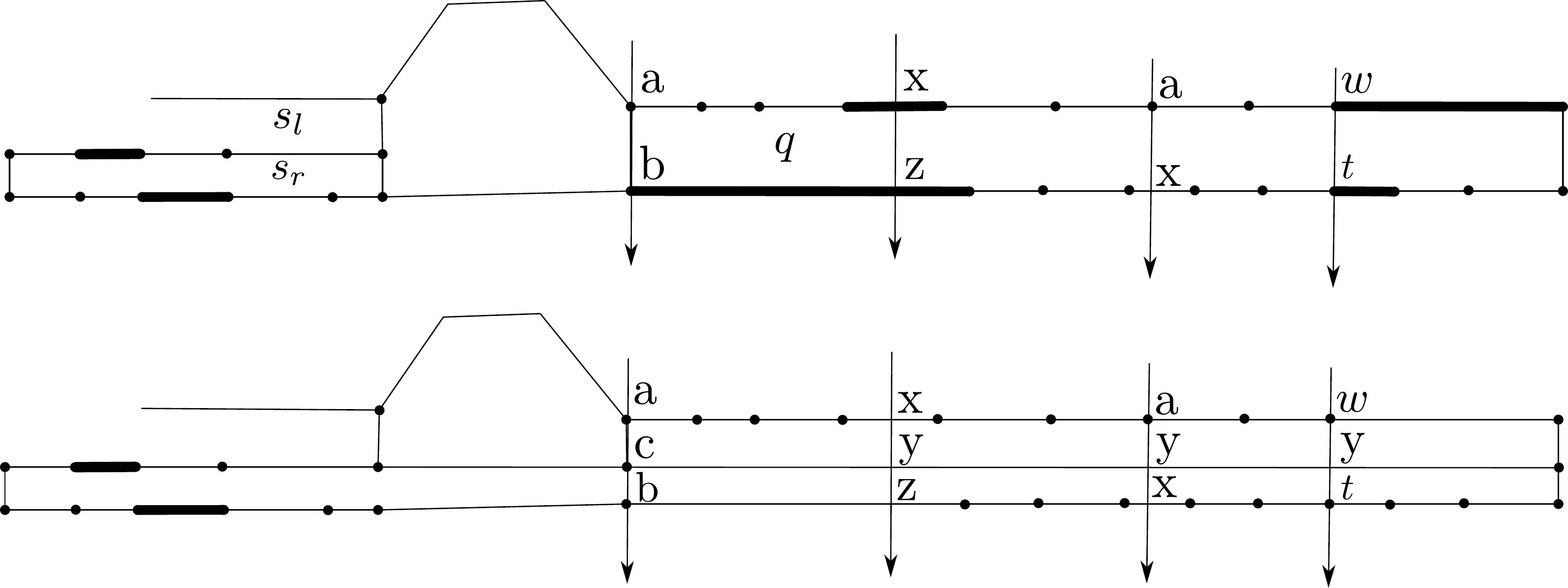}
    \caption{The effect of a right turn on the model streets, the arrows depict parts of the edge $e$, top: before the extension and bottom: after the extension of the street $s_r$ along the street $q$}
    \label{fig:modelextend}
\end{figure}

When traversing an edge $e=uv$ of $\gamma$ in the overlay of the street complex and the original complex $K$, we see a sequence of these terminals which we call the \textit{overlay sequence} of $e$. Our first goal is to construct a composition system for the overlay sequence.

The algorithm that traces the curves in $\Delta$ to build the street complex starts with the triangulation $K$, where each edge is a (trivial) street and there is one fork where we start the tracing. At the beginning, we introduce for the edge $e=uv$ of $\gamma$ a non-terminal $A(s_e, u,v)$, where $s_e$ is the street of the edge $e$. We also introduce the rule $E \rightarrow A(s_e,u,v)$. This finishes the modifications for the starting street complex. We note that the composition system is not a well-defined until the algorithm is finished.

Assume that we turn right at a junction and let $a$ and $b$ be vertices of the street $q$ which is going to be shrunk, and $s_r$ be the right active street, as in the figure \ref{fig:modelextend}. If the non-terminal $A(q,a,b)$ exists, we introduce the rule $A(q,a,b) \rightarrow A(q,a,c)A(s_r,c,b)$ where $c$ is as in the figure. If the street coordinates of $\Delta$ in (new) $q$ and/or in $s_r$ is 0, we also introduce terminals and the rule $A(q,a,c) \rightarrow (q,a,c)$ and/or $A(s_r, c, b) \rightarrow (s_r,c,b)$. Similarly, if the non-terminal $A(q,b,a)$ exists we introduce the rule $A(q,b,a) \rightarrow A(s_r,b,c)A(q,c,a).$
We introduce also the suitable terminals in case any of the streets have coordinates zero for $\Delta$.
Recall that during the tracing algorithm of~\cite{Erickson2013}, the number of arcs of $\Delta$ inside each street is maintained as the street coordinates. If a street has 0 number of arcs of $\Delta$, it will not be shrunk again, and in this case we introduce the leafs of composition system.

Next, for any two vertices $w,t$ on different sides of the model street for $q$, if the non-terminal $A(q,w,t)$ exists, and $w$ lies on the fixed side of $q$, we introduce the rule $A(q,w,t) \rightarrow A(q,w,y)A(s_r,y,t)$ where $y$ is the edge of (new) street $q$ as in the figure \ref{fig:modelextend}. If the street coordinates of $\Delta$ in $q$ and/or in $s_r$ is 0, we also introduce terminals and the rule $A(q,w,y) \rightarrow (q,w,y)$ and/or $A(s_r, y, t) \rightarrow (s_r,y,t)$. Similarly, if the non-terminal $A(q,w,t)$ exists where $t$ is on the fixed side of the street $q$ we introduce the rule $A(q,w,t) \rightarrow A(s_r,w,y)A(q,y,t).$ We introduce also the suitable terminals in case any of the streets have coordinates zero for $\Delta$. This finishes the processing of non-terminals of type iii).

If $A(q,x,z)$ is a non-terminal such that $x$ and $z$ lie on two sides of $q$ and $x$ lies on the fixed side of $q$, we introduce new non-terminals and the rule $A(q,x,z) \rightarrow A(q,x,y)A(s_r,y,z)$, where $y$ is the new side of the shrunk street $q$, as in the figure. We introduce new terminals if the coordinates are zero as above and do symmetrically if $z$ is on the fixed side of the street. This deals with the non-terminals of the second type. An analogous argument works for the non-terminals of the first type.

The modifications necessary for extension across a fork are totally analogous to a junction. In this case again a street is shrunk and another is extended, see~\cite{Erickson2013} figure 7. We do not repeat the case analysis.

We next consider a spiral. Let $d$ be the depth of the spiral, $l$ the length of the spiral, and $m$ be the number of distinct streets traversed by the spiral. See \cite{Erickson2013} for more detailed description of these parameters and how to compute the spiral. Let $\lambda$ be the long edge of the spiral, and $\lambda_1$, $\lambda_2$ be its preimages in the model. After we enter a spiral and before we reach the "final turn", the new side of the shrunk street is always $\lambda$ and the extended street has the edge $\lambda$ on both sides. See figure \ref{fig:spiralcases}. Let, in the system at the time we discover the spiral, $A(s_i,x,y)$ be the non-terminal where $s_i$ is a street traversed by the spiral, and $x$ and $y$ be edges on its two sides such that $x$ lies on the fixed side of $s_i$. After we trace the whole spiral, we need to add the rule $A(s_i,x,y) \rightarrow A(s_a,x,\lambda_*)A(s_a,\lambda_*,\lambda_*)^{d'} A(s_i, \lambda_*, y)$, where $s_a$ is the active street. The number $d'$ is either $d-1$ or $d-2$ depending on where the street lies around the spiral, in the example of figure~\ref{fig:spiralcases} it is $d-1$. In place of the star we place indices $1,2$ as is appropriate. If any street coordinate in the streets $s_i$ becomes zero  we introduce the rule $A(s_i,\lambda_*,x) \rightarrow (s_i, \lambda_*,x)$.
If the active street has coordinates zero for $\Delta$ we add the rule $A(s_a,\lambda_*,\lambda_*) = (s_a,\lambda_*,\lambda_*)$. We then extend the street $s_a$ through the streets that are traversed more than $d$ times. If $y$ is on the fixed side of the street $s_i$ we do analogously. We perform these modifications for every existing non-terminal of the form $A(s_i,x,y)$. This finished the modifications of the composition system for the non-terminals of type ii). Those of other two types are handled similarly.

\begin{figure}[ht]
    \centering
    \includegraphics[scale=0.4]{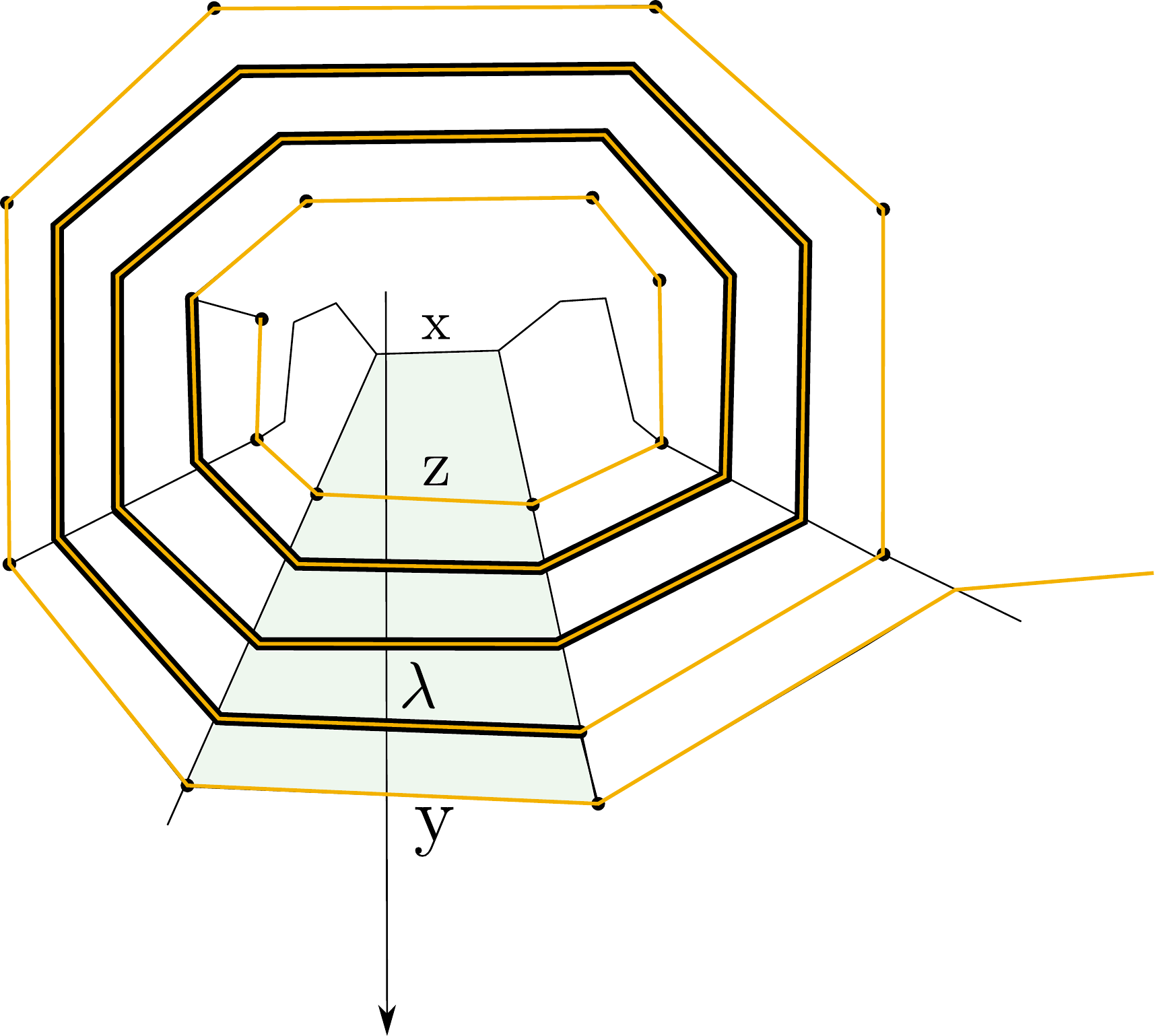}
    \caption{Updating the intersections of $e$ with a spiral, the orange curve is part of $\Delta$}
    \label{fig:spiralcases}
\end{figure}

We have finished the construction of the composition system $\mathbb{E}$. Note that since the curve $\Delta$ is reduced no street is an annulus. At this stage of the argument, we have replaced each rule $A \rightarrow e$ of the walk $\gamma$ by a composition system encoding its overlay sequence with the street complex. To construct $\gamma'$ we do as follows. For each edge of the the street complex, we choose one of the two endpoints. We next replace the terminals of the system $\mathbb{E}$ with a paths on the boundary of a street. This transforms the composition system into a composition system for the curve $\gamma'$ which is freely homotopic to $\gamma$, and $\gamma'$ is a walk on the 1-skeleton of the street complex.

Our next step is to triangulate the streets of the complex. For each street $s$, we triangulate the street by adding diagonals as in the Figure \ref{fig:streettr} a). Let $K'$ be a complex encoding this triangulation. $K'$ is simply the collection of our street models, junctions and their edges and vertices where we identify edges that are identified in the street complex. There is a homeomorphism $f:|K'| \rightarrow |K|$. By subdividing $K'$ we can make sure that it is a simplicial complex. We also modify the composition system for $\gamma'$ to reflect the subdivision.

\begin{figure}[ht]
    \centering
    \includegraphics[scale=0.45]{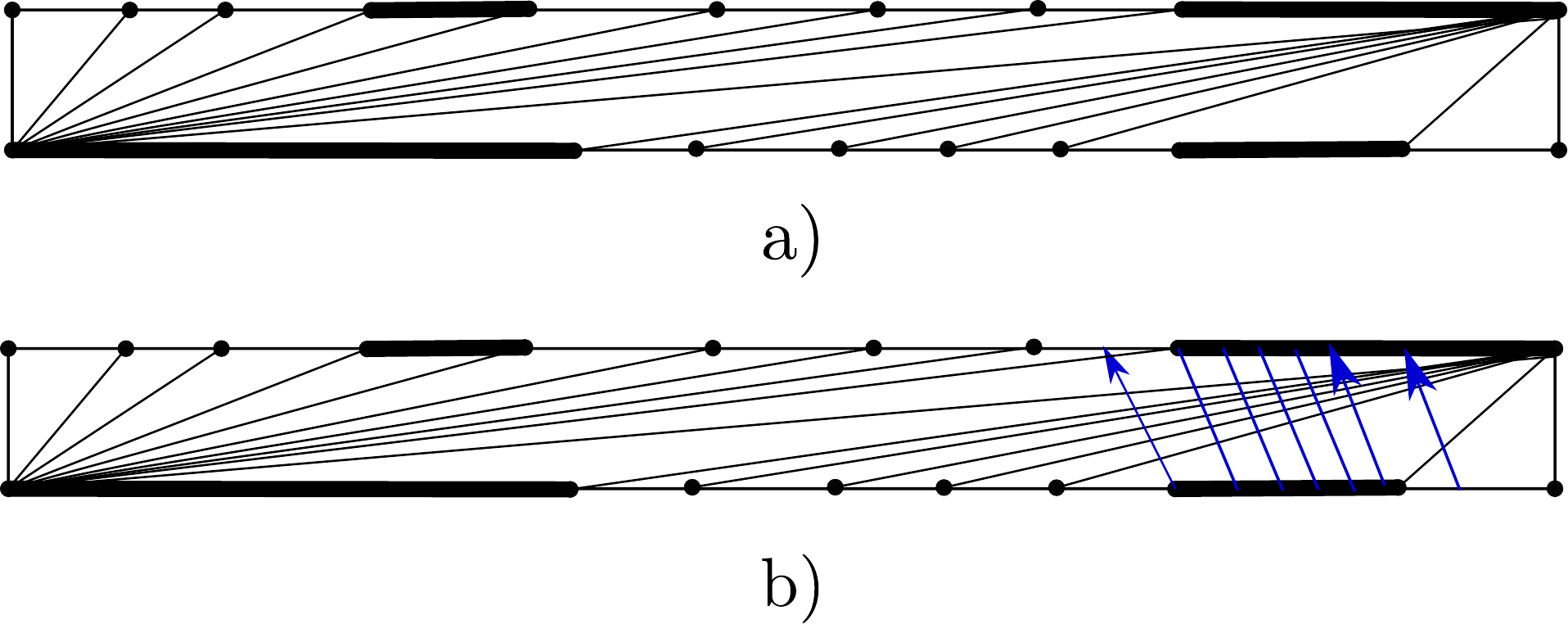}
    \caption{a) adding diagonals to triangulate a street, b) computing the intersection sequence of an arc with the triangulated street}
    \label{fig:streettr}
\end{figure}

\end{proof}

\subparagraph*{Remark}
The proof of Proposition~\ref{P:retriangulation} can also be used to construct the image $\gamma_1=f(\gamma)$ in the triangulation $K'$ as a hybrid compressed word whose terminals are either edges of the 1-skeleton, or arcs that connect a vertex to an edge or two edges in the same triangle. To compute $\gamma_1$ we start again with the street complex and the composition systems $\mathbb{E}$. We add the diagonals again as in Figure \ref{fig:streettr} a) to obtain a complex $K'$. Replace each terminal of the form $(s,x,y)$ where $x$ and $y$ are (not necessarily) distinct edges of the street $s$, by the sequence of intersections of the arc $xy$ with the diagonals of $s$. To find this sequence, we simply connect a points in the interior of $x$ to a point in the interior of $y$ by a direct line segment in our model street and replace $xy$ with the sequence of intersections obtained, as in figure \ref{fig:streettr} b). We do this process analogously for other types of arcs, yielding the curve $\gamma_1$. The curve $\gamma_1$ might be of independent interest for other problems.

\section{Capping off}\label{S:cappingoff}
In this section, we show that the complex obtained by gluing disks on a family of disjoint curves on a surface is homotopy equivalent to a wedge of surfaces and circles. We provide a polynomial-time algorithm to compute the resulting wedge, while also tracking what happens to a compressed curve on the original surface.

 \begin{proposition}\label{P:cappingoff}
 Let $L$ be simplicial complex triangulating a surface $S$ of genus $\geq 1$, $\Delta$ be a family of disjoint embedded closed curves in the 1-skeleton of $L$ and $\gamma$ be a compressed closed walk on the 1-skeleton of $L$. The size of the input is the summation of the complexities of $L$, $\gamma$ and the number of edges of $\Delta$. We can compute in polynomial time a simplicial complex $K$ which is a wedge of surfaces and circles, and a compressed walk $w$ on $K$, so that:
 \begin{itemize}
     \item $K$ is homotopy equivalent to the complex obtained by gluing disks on each component of $\Delta$,
     \item $w$ is homotopic to a trivial walk in $K$ if and only if $\gamma$ belongs to the normal subgroup determined by the curves in $\Delta$.
 \end{itemize}
 \end{proposition}

\begin{proof}
Note that the second property will be satisfied if, in the complex $K$ satisfying the first property, we take for $w$ the image of $\gamma$ under a homotopy equivalence of $L$ to $K$. We construct the simplicial complex $K$ which is homotopy equivalent to the space obtained from $L$ by gluing a disk to each component of $\Delta$. This is done one disk at a time, and at each step we also maintain the image of the curve $\gamma$ as a compressed word. Suppose $\Delta$ has $m$ components. As a preprocessing step, we subdivide any edge that connects vertices on $\Delta$ but does not lie on $\Delta$, twice,  subdivide the incident triangles, and then update the compressed word.

Call a space $X$ a \textit{pinched surface} if there is a finite set of vertices $P$ such that the closure of $X-P$ is a surface. We call $P$ the set of pinch points and assume $P$ is minimal, in the sense that $X$ is not homeomorphic to $\mathbb{R}^2$ in any neighborhood of a point in $P$.

We start by computing a straight-line program in Chomsky normal form for $\gamma$. During the procedure that processes components of $\Delta$, we maintain the following invariants:
i) the current space $|K'|$ is a pinched surface, ii) the current space $|K'|$ is homotopy equivalent to $|L|$ union caps over components of $\Delta$ that we have processed so far, iii) the current curve $w'$ is the image of the original curve $\gamma$ under a homotopy equivalence.

Let $c$ be a component of $\Delta$ and let $K'$ be the current complex. We add a disk with boundary $c$; we then compress this disk into a single vertex $x$ and add $x$ to $P$. See Figure~\ref{fig:contract}. The effect on the word $w'$ is as follows. A terminal $e=uv$ where $uv$ is an edge in $c$ is replaced with $x$. A terminal $e=uv$, where $v$ is in $c$ but $u$ is not, is replaced with $e'=ux$. If $u$ is in $c$ and $v$ is not, the oriented edge $e=uv$ is replaced with $e'=xv$. We then remove all the sequences $xx \cdots x$ in the word $w'$. This can be accomplished in polynomial time. We note that each remaining component of $\Delta$ is an embedded curve in the 1-skeleton lying in a single component of $|K'|-P$. The invariants remain true, and the word $w'$ is now a walk in the new complex $K'$.

\begin{figure}[ht]
    \centering
    \includegraphics[scale=0.24]{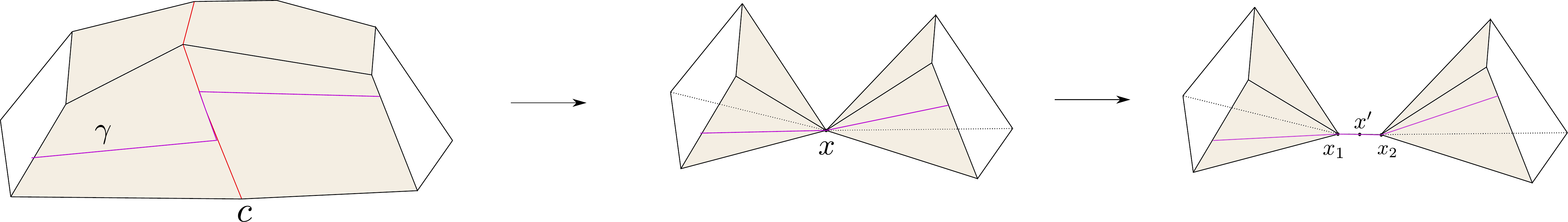}
    \caption{Contraction of a component $c$ of $\Delta$ (shown on the left) results in a new complex $K'$ (shown on the right).}
    \label{fig:contract}
\end{figure}

After contracting all the components of $\Delta$ we have a complex $K'$ which is a pinched surface. We transform $K'$ into a homotopy equivalent surface by replacing each pinch point $x$ by a path of two edges $x_1x'x_2$ where $x_1x'x_2$ maps to $x$ under the homotopy equivalence that contracts $x_1x'x_2.$ We then compute a spanning tree of the vertices $x_1,x_2$ for $x \in P$. We require that the number of $x'$ vertices is minimal in this tree. We then contract every edge of the spanning tree. The result is a complex $K$ which is wedge of circles and surfaces.

We explain the modifications needed to maintain the word $w'$. The first step is when we introduce the paths $x_1x'x_2$. Define $e_0=x_1x'$, $e_1=x'x_2$, and let $\bar{e_i}$ denote the edges with the opposite orientation. Observe that the link of the vertex $x$ consists of two disjoint circles. We associate $x_1$ to one of these and $x_2$ to the other arbitrarily. If a terminal is an edge $e=ux$, where $u$ is in the circle associated to $x_2$, we replace $e$ with the path $e' \bar{e_1}$, where $e'=ux_2$. If $u$ is in the circle associated to $x_1$, we replace $e$ with the path $e'e_0$ where $e'=ux_1$. We do analogously for the edges of the form $e=xv$. These modifications define a curve which is the image of the homotopy equivalence which introduces the paths.
It remains to update $w'$ under contraction of a single edge. This is possible in polynomial time similar to contraction of a circle above. The result is the required curve $w$.

\end{proof}

\section{Triviality for compressed words in free products of surface groups}

In this section, we prove the following theorem:

\begin{theorem}\label{T:triviality}
Let $K$ be a wedge of combinatorial surfaces and circles and $w$ be a compressed walk on $K$. Then one can compute in polynomial time a canonical form for the homotopy class of $w$. In particular, we can test in polynomial time whether $w$ is trivial.
\end{theorem}

We emphasize that in this theorem, we consider $w$ as a walk with fixed endpoints, and thus we are considering based homotopies (as opposed to free homotopies).

We can assume that there are no spheres in the wedge of surfaces: since these are simply connected, removing them, and replacing their edges by empty letters in the word $w$ does not change the homotopy class, or contractibility of $w$. Furthermore, tori require a slightly different treatment from the other surfaces. For the ease of exposition, we first explain and justify our algorithm assuming that there are no tori in the wedge, and at the end of the proof we will explain how to deal with them.

Our algorithm starts in a similar way as the known linear-time algorithms to test contractibility or homotopy of curves on surfaces~\cite{tcs,lazarusrivaud}: we first turn each surface in the wedge into a system of quads (Lemma~\ref{L:quadsystem}), and then compute a canonical form for our curve $w$ (Lemma~\ref{L:reduction}). This canonical form is unique, and therefore the walk $w$ is homotopic to a trivial walk if and only if its reduced canonical form is the empty walk. However, due to the compression of the input, our techniques to reduce the curve $w$ are more involved than those of the aforementioned references.

Let $\mathbb{A}$ denote the straight-line program encoding $w$. First, we fix once and for all an orientation for all the surfaces in $K$, which gives meaning to the intuitive notions of turning right, clockwise, etc.

 Then our first step is to turn our surfaces into systems of quads. First, for each surface $S$, we compute a spanning tree and contract it. Then, we remove edges until there is a single face, we add a new vertex inside the single face, add all the radial edges between this new vertex and the single vertex of $S$ and remove all the previous edges. To make things unified, in each circle summand, each loop is subdivided into two edges. The resulting complex is called a \textit{quad system}.

\begin{lemma}\label{L:quadsystem}
Let $K$ be a wedge of surfaces and circles and let $w$ be a compressed walk on $K$. In polynomial time we can turn $K$ into a quad system $K'$ and compute a compressed walk $w'$ on $K'$ that is homotopic to $w$.
\end{lemma}

\begin{proof}The construction of the quad system follows directly from the definition. The word $w$ is modified as follows: each terminal character representing an edge $e$ (which got removed) is now a non-terminal character with a production rule $e \rightarrow e_1e_2$ where $e_1e_2$ is a two-edge path homotopic to $e$ going through the new middle vertex. On surfaces, there can be two such edge-paths, we choose arbitrarily. The other production rules are unchanged. This operation can be done in polynomial time and yields a compressed curve homotopic to $w$.
\end{proof}

In the next step of our algorithm, we encode words using \textit{turn sequences}. For any two directed edges $e$ and $e'$ on the same surface $S$, the turn $\tau(e,e')$ is the number of corners, counted with respect to our chosen orientation, between the end of $e$ and the start of $e'$ on $S$. For any two edges $e$ and $e'$ forming a loop in a circle summand, we define $\tau(e,e')$ to be a new symbol $\lambda$, and for an edge $e$ in an circle summand, we define $\tau(e,e^{-1})$ to be $0$, which is in line with the case of surfaces. We need our encoding to adapt to jumps between different surfaces, and we do it as follows: If $e$ and $e'$ are on different surfaces or circle summands, we define the turn $\tau(e,e')$ to be the pair $ee'$, with the intended meaning that in this case the turn is a jump from $e$ to $e'$ on the new surface. Turn sequences are to be understood modulo the degree of the vertices, and following the literature, we use the notation $\bar{x}:=-x$. The symbol $\lambda$ satisfies $\bar{\lambda}=\lambda$. Henceforth, by a slight abuse of language, we count circle summands among the surface since they will be treated the same way.
The turn sequence of our word is the concatenation of all the turn sequences of consecutive edges.
However, we lose information by encoding with turn sequences. First, a turn sequence does not specify a starting edge. Second, even if we know the starting edge, it is not immediately clear how to compute an arbitrary  intermediate edge along the path efficiently since we encode exponential length paths. The following lemma shows how to do that in polynomial time.

\begin{lemma}\label{algo:second}
Let $\mathbb{T}$ be a compressed turn sequence and $k$ be an integer. Then if we are given a starting edge $e_{initial}$, then in polynomial time we can compute the edge of $K$ we are on just after starting at $e_{initial}$ and following the turn sequence up to $\mathbb{T}[k]$.
\end{lemma}

\begin{proof}
We first compute a straight-line program for $\mathbb{T}[:k]$, then use Lemma~\ref{algo:main} to compute the last letter of this word that corresponds to an edge. We denote it by $e_{last}$, and if there is none we use $e_{initial}$ instead. We denote by $i$ the index where it appears, which is $0$ if there is no such edge. Then we look at the straight-line program computing $\mathbb{T}[i:k]$, which thus stays on a single surface. We inductively go up the production tree, computing, for each directed edge $e$ of the surface and for each production rule at which directed edge we arrive if we start at edge $e$.  This is straightforward at the leaves of the tree, and for a directed edge $e$ and a production rule of $\mathbb{A}$ of the form $A_i \rightarrow A_jA_k$, we can look up where we arrive after starting at $e$ and following $A_j$, denote the resulting edge by $e'$, and then look up where we arrive after starting at $e'$ and following $A_j$. Finally, we obtain our solution by looking up where we are after starting at $e_{initial}$ and following $\mathbb{T}[i:k]$.
\end{proof}

 The point of turn sequences is that they make it easier to make local simplifications to a contractible curve. Following Lazarus-Rivaud~\cite{lazarusrivaud} and Erickson-Whittlesey~\cite{tcs}, we use the following simplification rules, see Figure~\ref{F:spursbrackets}, which are homotopies.

 \begin{figure}[t]
 \centering
 \def\svgwidth{\textwidth}
\begingroup%
  \makeatletter%
  \providecommand\color[2][]{%
    \errmessage{(Inkscape) Color is used for the text in Inkscape, but the package 'color.sty' is not loaded}%
    \renewcommand\color[2][]{}%
  }%
  \providecommand\transparent[1]{%
    \errmessage{(Inkscape) Transparency is used (non-zero) for the text in Inkscape, but the package 'transparent.sty' is not loaded}%
    \renewcommand\transparent[1]{}%
  }%
  \providecommand\rotatebox[2]{#2}%
  \ifx\svgwidth\undefined%
    \setlength{\unitlength}{3681.29004551bp}%
    \ifx\svgscale\undefined%
      \relax%
    \else%
      \setlength{\unitlength}{\unitlength * \real{\svgscale}}%
    \fi%
  \else%
    \setlength{\unitlength}{\svgwidth}%
  \fi%
  \global\let\svgwidth\undefined%
  \global\let\svgscale\undefined%
  \makeatother%
  \begin{picture}(1,0.47437898)%
    \put(0,0){\includegraphics[width=\unitlength,page=1]{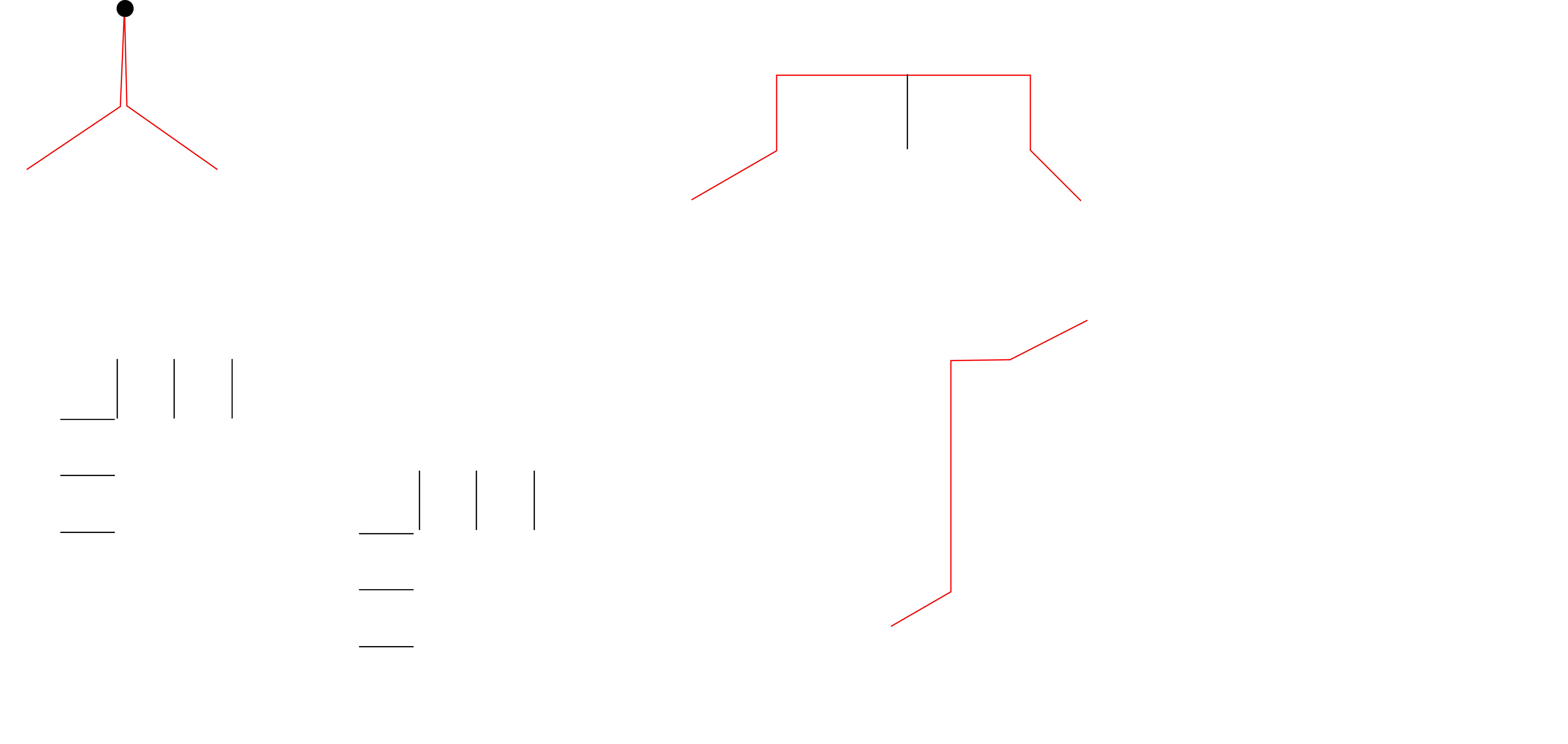}}%
    \put(0.04645779,0.40446905){\color[rgb]{0,0,0}\makebox(0,0)[lb]{\smash{$x$}}}%
    \put(0.05084821,0.46373984){\color[rgb]{0,0,0}\makebox(0,0)[lb]{\smash{$0$}}}%
    \put(0.09519156,0.4049081){\color[rgb]{0,0,0}\makebox(0,0)[lb]{\smash{$y$}}}%
    \put(0,0){\includegraphics[width=\unitlength,page=2]{spursbrackets.pdf}}%
    \put(0.24369105,0.42311813){\color[rgb]{0,0,0}\makebox(0,0)[lb]{\smash{$x+y$}}}%
    \put(0,0){\includegraphics[width=\unitlength,page=3]{spursbrackets.pdf}}%
    \put(0.46618275,0.37988265){\color[rgb]{0,0,0}\makebox(0,0)[lb]{\smash{$x$}}}%
    \put(0.67341097,0.38119978){\color[rgb]{0,0,0}\makebox(0,0)[lb]{\smash{$y$}}}%
    \put(0.49296436,0.4408976){\color[rgb]{0,0,0}\makebox(0,0)[lb]{\smash{$\overline{1}$}}}%
    \put(0.65490526,0.4408976){\color[rgb]{0,0,0}\makebox(0,0)[lb]{\smash{$\overline{1}$}}}%
    \put(0.53147389,0.4408976){\color[rgb]{0,0,0}\makebox(0,0)[lb]{\smash{$\overline{2}$}}}%
    \put(0.57173959,0.4408976){\color[rgb]{0,0,0}\makebox(0,0)[lb]{\smash{$\overline{2}$}}}%
    \put(0.61376147,0.4408976){\color[rgb]{0,0,0}\makebox(0,0)[lb]{\smash{$\overline{2}$}}}%
    \put(0.72094316,0.38031512){\color[rgb]{0,0,0}\makebox(0,0)[lb]{\smash{$x+1$}}}%
    \put(0.97346299,0.38012478){\color[rgb]{0,0,0}\makebox(0,0)[lb]{\smash{$y+1$}}}%
    \put(0.8213684,0.34386489){\color[rgb]{0,0,0}\makebox(0,0)[lb]{\smash{$2$}}}%
    \put(0.86607324,0.34386489){\color[rgb]{0,0,0}\makebox(0,0)[lb]{\smash{$2$}}}%
    \put(0.90425863,0.34386489){\color[rgb]{0,0,0}\makebox(0,0)[lb]{\smash{$2$}}}%
    \put(0,0){\includegraphics[width=\unitlength,page=4]{spursbrackets.pdf}}%
    \put(-0.04627974,0.2004153){\color[rgb]{0,0,0}\makebox(0,0)[lb]{\smash{}}}%
    \put(0.00867579,0.12982653){\color[rgb]{0,0,0}\makebox(0,0)[lb]{\smash{$\overline{2}$}}}%
    \put(0.00867579,0.16956414){\color[rgb]{0,0,0}\makebox(0,0)[lb]{\smash{$\overline{2}$}}}%
    \put(0.00867579,0.20681816){\color[rgb]{0,0,0}\makebox(0,0)[lb]{\smash{$\overline{2}$}}}%
    \put(0.06890311,0.25680065){\color[rgb]{0,0,0}\makebox(0,0)[lb]{\smash{$\overline{2}$}}}%
    \put(0.10739893,0.25680065){\color[rgb]{0,0,0}\makebox(0,0)[lb]{\smash{$\overline{2}$}}}%
    \put(0.14403205,0.25680065){\color[rgb]{0,0,0}\makebox(0,0)[lb]{\smash{$\overline{2}$}}}%
    \put(0.00867579,0.25680065){\color[rgb]{0,0,0}\makebox(0,0)[lb]{\smash{$\overline{1}$}}}%
    \put(0.30422436,0.10468008){\color[rgb]{0,0,0}\makebox(0,0)[lb]{\smash{$2$}}}%
    \put(0.34272017,0.10468008){\color[rgb]{0,0,0}\makebox(0,0)[lb]{\smash{$2$}}}%
    \put(0.37935329,0.10468008){\color[rgb]{0,0,0}\makebox(0,0)[lb]{\smash{$2$}}}%
    \put(0.28125105,0.08108584){\color[rgb]{0,0,0}\makebox(0,0)[lb]{\smash{$2$}}}%
    \put(0.28125105,0.04414227){\color[rgb]{0,0,0}\makebox(0,0)[lb]{\smash{$2$}}}%
    \put(0.28125105,0.00440464){\color[rgb]{0,0,0}\makebox(0,0)[lb]{\smash{$1$}}}%
    \put(0.28125105,0.10585139){\color[rgb]{0,0,0}\makebox(0,0)[lb]{\smash{$3$}}}%
    \put(0.38760918,0.13170432){\color[rgb]{0,0,0}\makebox(0,0)[lb]{\smash{$1$}}}%
    \put(0.00765619,0.09516668){\color[rgb]{0,0,0}\makebox(0,0)[lb]{\smash{$x$}}}%
    \put(0.17976102,0.2571735){\color[rgb]{0,0,0}\makebox(0,0)[lb]{\smash{$y$}}}%
    \put(0.15516389,0.02272461){\color[rgb]{0,0,0}\makebox(0,0)[lb]{\smash{$x+1$}}}%
    \put(0.40041954,0.16056447){\color[rgb]{0,0,0}\makebox(0,0)[lb]{\smash{$y+1$}}}%
    \put(0,0){\includegraphics[width=\unitlength,page=5]{spursbrackets.pdf}}%
    \put(0.57767533,0.12982653){\color[rgb]{0,0,0}\makebox(0,0)[lb]{\smash{$\overline{2}$}}}%
    \put(0.57767533,0.16956414){\color[rgb]{0,0,0}\makebox(0,0)[lb]{\smash{$\overline{2}$}}}%
    \put(0.57767533,0.20681816){\color[rgb]{0,0,0}\makebox(0,0)[lb]{\smash{$\overline{2}$}}}%
    \put(0.63351224,0.26031299){\color[rgb]{0,0,0}\makebox(0,0)[lb]{\smash{$y$}}}%
    \put(0.57767533,0.25680065){\color[rgb]{0,0,0}\makebox(0,0)[lb]{\smash{$\overline{1}$}}}%
    \put(0.57665574,0.09516668){\color[rgb]{0,0,0}\makebox(0,0)[lb]{\smash{$x$}}}%
    \put(0,0){\includegraphics[width=\unitlength,page=6]{spursbrackets.pdf}}%
    \put(0.7792494,0.04047352){\color[rgb]{0,0,0}\makebox(0,0)[lb]{\smash{$x+1$}}}%
    \put(0,0){\includegraphics[width=\unitlength,page=7]{spursbrackets.pdf}}%
    \put(0.9067585,0.14612819){\color[rgb]{0,0,0}\makebox(0,0)[lb]{\smash{$2$}}}%
    \put(0.9067585,0.10918459){\color[rgb]{0,0,0}\makebox(0,0)[lb]{\smash{$2$}}}%
    \put(0.9067585,0.07224102){\color[rgb]{0,0,0}\makebox(0,0)[lb]{\smash{$2$}}}%
    \put(0.9067585,0.03250339){\color[rgb]{0,0,0}\makebox(0,0)[lb]{\smash{$1$}}}%
    \put(0.91301129,0.18042705){\color[rgb]{0,0,0}\makebox(0,0)[lb]{\smash{$y+1$}}}%
    \put(0,0){\includegraphics[width=\unitlength,page=8]{spursbrackets.pdf}}%
  \end{picture}%
\endgroup%

 \caption{The reductions to eliminate a spur (top left) and a bracket (top right), and the shifting moves to remove a $\overline{1}$ (bottom row).}
 \label{F:spursbrackets}
 \end{figure}

 \begin{itemize}
     \item A \textit{spur} in a turn sequence $w$ is a $0$ turn. Removing a spur is applying the rule $x0y\rightarrow x+y$.
     \item A \textit{bracket} in a turn sequence $w$ is a subword of the form $12\ldots 21$ or $\bar{1}\bar{2}\ldots \bar{2}\bar{1}$.
     Flattening a bracket is applying the rules $x12\ldots 21y \rightarrow (x-1)\bar{2}\ldots \bar{2}(y-1)$ or $x\bar{1}\bar{2}\ldots \bar{2}\bar{1}y \rightarrow (x+1)2\ldots 2(y+1)$.
 \end{itemize}

 We are only considering homotopies of paths in this paper, and thus do not need the other special cases considered in Erickson-Whittlesey~\cite{tcs}.
Our goal is to modify the turn sequence so that it is \textit{reduced}, i.e., contains neither spurs nor brackets. However, unlike in the aforementioned works, we cannot apply these rules directly, even inductively: in a production rule $A_i \rightarrow A_jA_k$, even when $A_j$ and $A_k$ are reduced, there might be an exponential number of bracket flattenings in $A_i$, for example in the cases in Figure~\ref{F:snakes}.

\begin{figure}
    \centering
    \def\svgwidth{\textwidth-4cm}
    \begingroup%
  \makeatletter%
  \providecommand\color[2][]{%
    \errmessage{(Inkscape) Color is used for the text in Inkscape, but the package 'color.sty' is not loaded}%
    \renewcommand\color[2][]{}%
  }%
  \providecommand\transparent[1]{%
    \errmessage{(Inkscape) Transparency is used (non-zero) for the text in Inkscape, but the package 'transparent.sty' is not loaded}%
    \renewcommand\transparent[1]{}%
  }%
  \providecommand\rotatebox[2]{#2}%
  \ifx\svgwidth\undefined%
    \setlength{\unitlength}{1706.54042969bp}%
    \ifx\svgscale\undefined%
      \relax%
    \else%
      \setlength{\unitlength}{\unitlength * \real{\svgscale}}%
    \fi%
  \else%
    \setlength{\unitlength}{\svgwidth}%
  \fi%
  \global\let\svgwidth\undefined%
  \global\let\svgscale\undefined%
  \makeatother%
  \begin{picture}(1,0.72088604)%
    \put(0,0){\includegraphics[width=\unitlength,page=1]{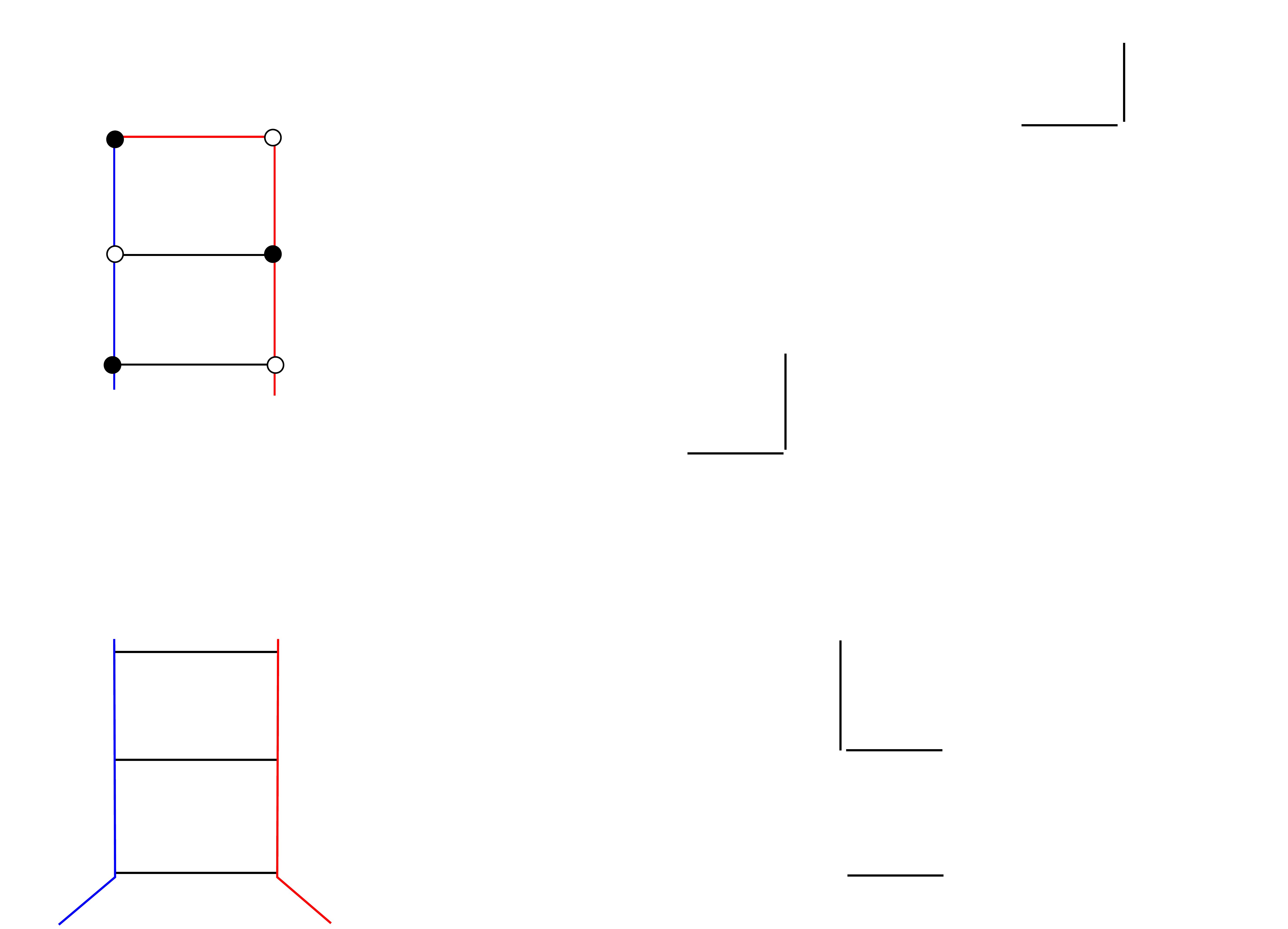}}%
    \put(0.08471938,0.31340546){\color[rgb]{0,0,0}\makebox(0,0)[b]{\smash{$\vdots$}}}%
    \put(0.21263062,0.31340546){\color[rgb]{0,0,0}\makebox(0,0)[b]{\smash{$\vdots$}}}%
    \put(0,0){\includegraphics[width=\unitlength,page=2]{snakes.pdf}}%
    \put(0.01487567,0.12521174){\color[rgb]{0,0,0}\makebox(0,0)[lb]{\smash{$2$}}}%
    \put(0.01487567,0.2076084){\color[rgb]{0,0,0}\makebox(0,0)[lb]{\smash{$2$}}}%
    \put(0.01487567,0.43112118){\color[rgb]{0,0,0}\makebox(0,0)[lb]{\smash{$2$}}}%
    \put(0.01487567,0.5173062){\color[rgb]{0,0,0}\makebox(0,0)[lb]{\smash{$2$}}}%
    \put(0.25259477,0.12521174){\color[rgb]{0,0,0}\makebox(0,0)[lb]{\smash{$2$}}}%
    \put(0.25259477,0.2076084){\color[rgb]{0,0,0}\makebox(0,0)[lb]{\smash{$2$}}}%
    \put(0.25259477,0.43112118){\color[rgb]{0,0,0}\makebox(0,0)[lb]{\smash{$2$}}}%
    \put(0.25259477,0.5173062){\color[rgb]{0,0,0}\makebox(0,0)[lb]{\smash{$2$}}}%
    \put(0.01593771,0.60725987){\color[rgb]{0,0,0}\makebox(0,0)[lb]{\smash{$1$}}}%
    \put(0.24991602,0.60677237){\color[rgb]{0,0,0}\makebox(0,0)[lb]{\smash{$1$}}}%
    \put(0,0){\includegraphics[width=\unitlength,page=3]{snakes.pdf}}%
    \put(-0.00122484,0.03287066){\color[rgb]{0,0,1}\makebox(0,0)[lb]{\smash{$A_k$}}}%
    \put(0.26680108,0.03287066){\color[rgb]{1,0,0}\makebox(0,0)[lb]{\smash{$A_j$}}}%
    \put(0.30373756,0.03902672){\color[rgb]{0,0,0}\makebox(0,0)[lb]{\smash{}}}%
    \put(0.54817583,0.04238915){\color[rgb]{0,0,1}\makebox(0,0)[lb]{\smash{$A_k$}}}%
    \put(0.77828798,0.04286269){\color[rgb]{1,0,0}\makebox(0,0)[lb]{\smash{$A_j$}}}%
    \put(0,0){\includegraphics[width=\unitlength,page=4]{snakes.pdf}}%
    \put(0.75190497,0.48141204){\color[rgb]{0,0,0}\makebox(0,0)[b]{\smash{$\iddots$}}}%
    \put(0,0){\includegraphics[width=\unitlength,page=5]{snakes.pdf}}%
    \put(0.6561851,0.23725999){\color[rgb]{0,0,0}\makebox(0,0)[lb]{\smash{$2$}}}%
    \put(0.75549854,0.1061134){\color[rgb]{0,0,0}\makebox(0,0)[lb]{\smash{$2$}}}%
    \put(0.47465913,0.36340438){\color[rgb]{0,0,0}\makebox(0,0)[lb]{\smash{$2$}}}%
    \put(0.73812075,0.6216205){\color[rgb]{0,0,0}\makebox(0,0)[lb]{\smash{$2$}}}%
    \put(0.84726408,0.7075298){\color[rgb]{0,0,0}\makebox(0,0)[lb]{\smash{$2$}}}%
    \put(0.73812075,0.7075298){\color[rgb]{0,0,0}\makebox(0,0)[lb]{\smash{$1$}}}%
    \put(0.96788836,0.7075298){\color[rgb]{0,0,0}\makebox(0,0)[lb]{\smash{$1$}}}%
    \put(0.96788836,0.58907147){\color[rgb]{0,0,0}\makebox(0,0)[lb]{\smash{$1$}}}%
    \put(0.89029266,0.52438051){\color[rgb]{0,0,0}\makebox(0,0)[lb]{\smash{$1$}}}%
    \put(0.47465913,0.46782155){\color[rgb]{0,0,0}\makebox(0,0)[lb]{\smash{$1$}}}%
    \put(0.72372802,0.3420253){\color[rgb]{0,0,0}\makebox(0,0)[lb]{\smash{$1$}}}%
    \put(0.47683078,0.1061134){\color[rgb]{0,0,0}\makebox(0,0)[lb]{\smash{$1$}}}%
    \put(0.77078273,0.23725999){\color[rgb]{0,0,0}\makebox(0,0)[lb]{\smash{$1$}}}%
    \put(0.61582772,0.1061134){\color[rgb]{0,0,0}\makebox(0,0)[lb]{\smash{$3$}}}%
    \put(0.62490629,0.3420253){\color[rgb]{0,0,0}\makebox(0,0)[lb]{\smash{$3$}}}%
    \put(0.89029266,0.58907147){\color[rgb]{0,0,0}\makebox(0,0)[lb]{\smash{$3$}}}%
    \put(0.47553308,0.23725999){\color[rgb]{0,0,0}\makebox(0,0)[lb]{\smash{$k$}}}%
    \put(0.56135708,0.23740162){\color[rgb]{0,0,0}\makebox(0,0)[lb]{\smash{$\overline{k}+1$}}}%
    \put(0.47273638,0.30083914){\color[rgb]{0,0,0}\makebox(0,0)[lb]{\smash{$k'$}}}%
    \put(0.55192559,0.26805209){\color[rgb]{0,0,0}\makebox(0,0)[lb]{\smash{$\overline{k'}+1$}}}%
    \put(0,0){\includegraphics[width=\unitlength,page=6]{snakes.pdf}}%
    \put(0.58502267,0.46782155){\color[rgb]{0,0,0}\makebox(0,0)[lb]{\smash{$2$}}}%
    \put(0.6257044,0.29765935){\color[rgb]{0,0,0}\makebox(0,0)[lb]{\smash{$1$}}}%
  \end{picture}%
\endgroup%

    \caption{Some production rules $A_i \rightarrow A_jA_k$ require an exponential number of bracket flattenings and/or spur removals to be reduced, despite $A_j$ and $A_k$ being already reduced.}
    \label{F:snakes}
\end{figure}

There are known techniques to handle exponentially long spurs~\cite{lohrey,Schleimer2008}, but for the more intricate cases pictured in Figure~\ref{F:snakes}, especially the kind on the right side, we need to develop our own tools. In order to do that, we rely on stronger inductive forms, similarly to those used in the free homotopy test~\cite{tcs,lazarusrivaud}. A reduced path is \textit{leftmost} (or \textit{rightmost}, respectively) if its turn sequence contains no $1$, respectively no $\bar{1}$ (we emphasize that we also use this definition in our general setting of wedges of surfaces).
Note that leftmost paths become rightmost paths under reversal, and vice-versa. This trivial observation will fuel many of our arguments. One can transform a reduced path into its rightmost (resp. leftmost) form by doing \textit{elementary right-shifts} (resp. \textit{elementary left-shifts}) which are the three transformations $x\overline{2}^s\overline{1}\overline{2}^ty \rightarrow (x+1)12^{s-1}32^{t-1}1(y+1)$, $x\overline{2}^s\overline{1}y \rightarrow (x+1)12^s(y+1)$ and $x\overline{1}2^ty\rightarrow (x+1)2^t1(y+1)$ or their mirrors, see Figure~\ref{F:spursbrackets}.

For each character in the straight-line program, we inductively compute both a rightmost and leftmost turn sequence. Then, both are used to carry the induction step. So the next step of our algorithm is the following reduction algorithm. It takes as input the straight-line program $\mathbb{A}$ describing a walk on a system of quads, and at the same time transforms it into two compressed turn sequences representing the reduced leftmost and rightmost form for the walk. In these representations, every character $A_i$ of $\mathbb{A}$ encoding a walk $w_i$ gives rise to two characters $A^L_i$ and $A^R_i$ representing the turn sequences of the leftmost and rightmost form of $w_i$. Since turn sequences only encode the turns at the interior vertices of a path and forget where the path starts, we store this information separately: for each character in the straight-line program, we store in a dictionary its starting and ending edges (which might be empty).

\vspace{1em}

\textsc{Reduction Algorithm}

\textbf{Input}: A compressed walk on a quad system described by a straight-line program $\mathbb{A}$, as output by Lemma~\ref{L:quadsystem}.

\textbf{Output}: Two compressed turn sequences on a quad system which are the leftmost and rightmost forms of the input, and their starting and ending edges.

\begin{itemize}
    \item The leaves of the production tree are transformed to empty sequences since they consist of a single edge. We store the edge as both the starting and the ending edge in our dictionary.
    \item Let $A_i \rightarrow A_jA_k$ be a production rule, and assume that we have inductively computed rightmost and leftmost reduced turn sequences for $A_j$ and $A_k$, which are denoted by $A^R_j$, $A^L_j$, $A^R_k$ and $A^L_k$, as well as their starting and ending edges. We explain how to compute a rightmost reduced turn sequence $A_i^R$, the leftmost case being symmetric. If the dictionary entries of $A^R_j$ and $A^R_k$ are empty, $A^R_i$ is an empty sequence, with empty starting and ending edges. If one of the dictionary entries is empty, $A^R_i$ is the other turn sequence, inheriting its starting and ending edges. Otherwise:
    \begin{enumerate}
        \item We look up the last edge of $A_j^R$ and the first edge of $A_k^R$ in the dictionary and compute the turn $\tau$ between them \footnote{At the start of steps 1, 2 and 3, it might happen that the last letter of $A_j^R$ and the first letter of $A_k^R$ are both \textit{edges}, if half of a jumping turn got removed. In that case, we remove those edges since they are superfluous.}. If it is different from $0$ or $1$, we go to step 3. Otherwise, we use Lemma~\ref{algo:main} to compute the largest $m$ so that $A_j^R[-m:]=\overline{A_k^L}[-m:]$ if they have the same ending edge (let us emphasize that we use the \textbf{leftmost} form here). We use Lemma~\ref{algo:second} to compute the last edge of $\overline{A_k^R}[:-m-1]$. If it arrives at the same vertex as $\overline{A_k^L}[:-m-1]$, we go to step 2 with $A_j^R:=A_j^R[:-m-1]$ and $A_k^R:=A_k^R[m+1:]$. If not, we go to step 2 with $A_j^R:=A_j^R[:-m-1]$ and $A_k^R:=(\tau-1) A_k^R[m+2:]$, where $\tau$ is the turn $A_k^R[m+1]$.

        \item Using Lemma~\ref{algo:second}, we compute\footnotemark[\value{footnote}] the turn $\tau$ between the last edge of $A_j^R$ and the first edge of $A_k^R$. If $\tau=1$, we use Lemma~\ref{algo:main} to compute a maximal subword $x2^{m'}112^{m'}y$ in $A_j^R\tau A_k^R$, i.e.,  $A_j^R\tau A_k^R=Sx2^{m'}112^{m'}yT$, and we set $A_j^R$ to be $S(x-1)$ and $A_k^R$ to be $T$ and go to step 3.
        It $\tau=0$, we use Lemma~\ref{algo:main} to compute the largest $m$ so that $A_j^R[-m:]=\overline{A_k^R}[-m:]$ and go to step 3 with $A_j^R:=A_j^R[:-m-1]$ and $A_k^R:=A_k^R[m+1:]$.
        \item If we came directly here, we look up the last edge of $A_j^R$ and the first edge of $A_k^R$ in the dictionary and compute the turn $\tau$ between them\footnotemark[\value{footnote}]. If we came from step 2, we compute\footnotemark[\value{footnote}] these edges using Lemma~\ref{algo:second} and use them to compute $\tau$. We concatenate the two words, adding $\tau$ in between.
        \begin{enumerate}
            \item[a.] If $\tau=\overline{1}$, we use Lemma~\ref{algo:main} to compute the maximal subword $x\overline{2}^s\overline{1}\overline{2}^ty$, and apply an elementary shift which changes it into $(x+1)12^{s-1}32^{t-1}1(y+1)$,  $(x+1)12^s(y+1)$, $(x+1)2^t1(y+1)$, or  depending on whether $s$ or $t$ is zero (if $x$ or $y$ are edges, they are changed into the turned edges). Since it is straightforward to encode efficiently $2^s$ using straight-line programs, this is done in polynomial time. We go to step b.
            \item[b.] If $\tau=1$ or $\tau=2$, we use Lemma~\ref{algo:main} to compute a maximal subword of the form $x12^k1y$. There might be multiple such subwords, we pick one of them arbitrarily. We change it into $(x-1)\overline{2}^k(y-1)$ (if $x$ or $y$ are edges, they are changed into the turned edges). Since it is straightforward to encode efficiently $2^k$ using straight-line programs, this is done in polynomial time. We loop this step while there are brackets. When there are none, we output the result and store in our dictionary the first and last edge: these are the first and last edge of respectively, $A_j^R$ and $A_k^R$, perhaps turned because of one of the elementary operations, or perhaps made empty if the entire word became empty.
            \item[c.] In any other case, we do not change anything and output the concatenated word, and store in our dictionary the first edge of $A_j^R$ and the last edge of $A_k^R$.
        \end{enumerate}

    \end{enumerate}
\end{itemize}

Step 3 simply implements the bracket flattenings and the shifts, so the main mysteries in this algorithm are steps 1 and 2. These are tailored to deal with the bad cases pictured in Figure~\ref{F:snakes}, which are the only possible bad cases; see the proof of Lemma~\ref{L:removingstaircases}. Figure~\ref{F:reductions} illustrates how step 1 reduces (a subset of) the right case of Figure~\ref{F:snakes} to its left case, and how step 2 deals with that case. Note that in these pictures, $A_j$ and $A_k$ were already rightmost.

\begin{figure}[ht]
    \def\svgwidth{14.5cm}
    \begingroup%
  \makeatletter%
  \providecommand\color[2][]{%
    \errmessage{(Inkscape) Color is used for the text in Inkscape, but the package 'color.sty' is not loaded}%
    \renewcommand\color[2][]{}%
  }%
  \providecommand\transparent[1]{%
    \errmessage{(Inkscape) Transparency is used (non-zero) for the text in Inkscape, but the package 'transparent.sty' is not loaded}%
    \renewcommand\transparent[1]{}%
  }%
  \providecommand\rotatebox[2]{#2}%
  \ifx\svgwidth\undefined%
    \setlength{\unitlength}{2545.68024687bp}%
    \ifx\svgscale\undefined%
      \relax%
    \else%
      \setlength{\unitlength}{\unitlength * \real{\svgscale}}%
    \fi%
  \else%
    \setlength{\unitlength}{\svgwidth}%
  \fi%
  \global\let\svgwidth\undefined%
  \global\let\svgscale\undefined%
  \makeatother%
  \begin{picture}(1,0.78971411)%
    \put(-0.20453753,0.54130644){\color[rgb]{0,0,0}\makebox(0,0)[lb]{\smash{}}}%
    \put(0,0){\includegraphics[width=\unitlength,page=1]{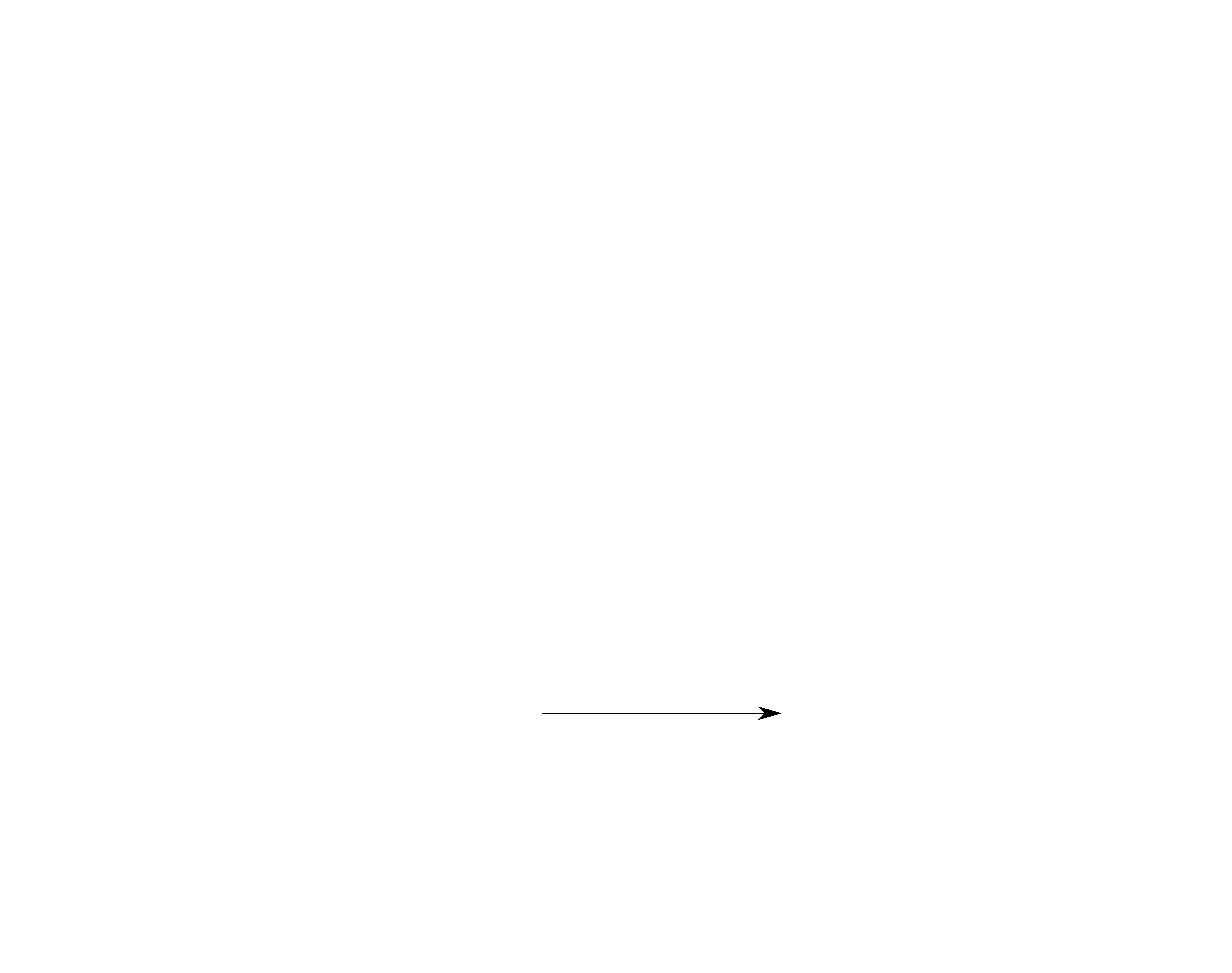}}%
    \put(0.518648,0.23140719){\color[rgb]{0,0,0}\makebox(0,0)[lb]{\smash{Step 2}}}%
    \put(0,0){\includegraphics[width=\unitlength,page=2]{reductions.pdf}}%
    \put(-0.00077682,0.37870324){\color[rgb]{0,0,0}\makebox(0,0)[lb]{\smash{Step 1}}}%
    \put(0,0){\includegraphics[width=\unitlength,page=3]{reductions.pdf}}%
    \put(0.05300266,0.46054731){\color[rgb]{0,0,1}\makebox(0,0)[lb]{\smash{$A^R_k$}}}%
    \put(0.2853778,0.46266057){\color[rgb]{1,0,0}\makebox(0,0)[lb]{\smash{$A^R_j$}}}%
    \put(0,0){\includegraphics[width=\unitlength,page=4]{reductions.pdf}}%
    \put(0.16850675,0.62619961){\color[rgb]{0,0,0}\makebox(0,0)[lb]{\smash{$2$}}}%
    \put(0.23508323,0.53828328){\color[rgb]{0,0,0}\makebox(0,0)[lb]{\smash{$2$}}}%
    \put(0.04681769,0.71076269){\color[rgb]{0,0,0}\makebox(0,0)[lb]{\smash{$2$}}}%
    \put(0.04681769,0.78076052){\color[rgb]{0,0,0}\makebox(0,0)[lb]{\smash{$1$}}}%
    \put(0.21378528,0.69643085){\color[rgb]{0,0,0}\makebox(0,0)[lb]{\smash{$1$}}}%
    \put(0.04827349,0.53828328){\color[rgb]{0,0,0}\makebox(0,0)[lb]{\smash{$1$}}}%
    \put(0.24532922,0.62619961){\color[rgb]{0,0,0}\makebox(0,0)[lb]{\smash{$1$}}}%
    \put(0.14145246,0.53828328){\color[rgb]{0,0,0}\makebox(0,0)[lb]{\smash{$3$}}}%
    \put(0.14753844,0.69643085){\color[rgb]{0,0,0}\makebox(0,0)[lb]{\smash{$3$}}}%
    \put(0.04740357,0.62619961){\color[rgb]{0,0,0}\makebox(0,0)[lb]{\smash{$k$}}}%
    \put(0.10493715,0.62058049){\color[rgb]{0,0,0}\makebox(0,0)[lb]{\smash{\small{$\overline{k}+1$}}}}%
    \put(0.04552874,0.66882099){\color[rgb]{0,0,0}\makebox(0,0)[lb]{\smash{$k'$}}}%
    \put(0.09861459,0.6452544){\color[rgb]{0,0,0}\makebox(0,0)[lb]{\smash{\small{$\overline{k'}+1$}}}}%
    \put(0,0){\includegraphics[width=\unitlength,page=5]{reductions.pdf}}%
    \put(0.12080178,0.78076052){\color[rgb]{0,0,0}\makebox(0,0)[lb]{\smash{$2$}}}%
    \put(0.14807347,0.66668937){\color[rgb]{0,0,0}\makebox(0,0)[lb]{\smash{$1$}}}%
    \put(0,0){\includegraphics[width=\unitlength,page=6]{reductions.pdf}}%
    \put(0.40522231,0.45766839){\color[rgb]{0,0.50196078,0}\makebox(0,0)[lb]{\smash{$A^L_k$}}}%
    \put(0,0){\includegraphics[width=\unitlength,page=7]{reductions.pdf}}%
    \put(0.6117857,0.47482976){\color[rgb]{1,0,0}\makebox(0,0)[lb]{\smash{$A^R_j$}}}%
    \put(0,0){\includegraphics[width=\unitlength,page=8]{reductions.pdf}}%
    \put(0.90984091,0.46665082){\color[rgb]{1,0,0}\makebox(0,0)[lb]{\smash{$A^R_j[:-m-1]$}}}%
    \put(0.69011706,0.45457365){\color[rgb]{0,0,1}\makebox(0,0)[lb]{\smash{$A^R_k[m+2:]$}}}%
    \put(0,0){\includegraphics[width=\unitlength,page=9]{reductions.pdf}}%
    \put(0.81196785,0.57825968){\color[rgb]{0,0,0}\makebox(0,0)[lb]{\smash{$\tau$}}}%
    \put(0,0){\includegraphics[width=\unitlength,page=10]{reductions.pdf}}%
    \put(0.93007725,0.01150923){\color[rgb]{1,0,0}\makebox(0,0)[lb]{\smash{$S(x-1)$}}}%
    \put(0.69768188,0.00188915){\color[rgb]{0,0,1}\makebox(0,0)[lb]{\smash{$T$}}}%
    \put(0,0){\includegraphics[width=\unitlength,page=11]{reductions.pdf}}%
    \put(0.35654289,0.01044172){\color[rgb]{1,0,0}\makebox(0,0)[lb]{\smash{$A^R_j[:-m-1]$}}}%
    \put(0.08057696,0.00329742){\color[rgb]{0,0,1}\makebox(0,0)[lb]{\smash{$(\tau-1)A^R_k[m+2:]$}}}%
    \put(0,0){\includegraphics[width=\unitlength,page=12]{reductions.pdf}}%
    \put(0.18615754,0.1274028){\color[rgb]{0,0,0}\makebox(0,0)[lb]{\smash{$\tau-1$}}}%
    \put(0,0){\includegraphics[width=\unitlength,page=13]{reductions.pdf}}%
  \end{picture}%
\endgroup%

    \caption{The leftmost form $A^L_k$ forms a long spur with $A^R_j$. After reducing $A^R_k$ by the length of this spur, all of the staircases get removed, except possibly the last one. It disappears in step 2.}
    \label{F:reductions}
\end{figure}

We claim that after this procedure, the path that represents the word $A_i^R$ (respectively $A_i^L$) is homotopic to the path represented by $A_i$, and is rightmost (respectively leftmost). This is straightforward for the leaves of the production tree, but the rest requires some more involved analysis.

For this analysis, we employ the terminology developed by Despré and Lazarus in their analysis of the combinatorics of quad systems~\cite[Section~4]{desprelazarus}, to which we refer the reader for definitions. We will rely on two important structural results~\cite[Theorem~8 and Theorem~10]{desprelazarus}(see also ~\cite{tcs,lazarusrivaud}).

\begin{theorem}\label{T:uniqueness} For a path $p$ on a combinatorial surface of negative Euler characteristic described by a system of quads, there is a unique rightmost, respectively leftmost, path homotopic to $p$.
\end{theorem}

\begin{theorem}\label{T:geodesics}
Let $c$ and $d$ be a rightmost and a leftmost path on a combinatorial surface of negative Euler characteristic described by a system of quads. Then $c \cdot d^{-1}$ bounds a disk diagram composed of an alternating sequence of paths and quad staircases connected through their tips.
\end{theorem}

The two theorems extend directly to our setting of wedge of surfaces and circles, since the fundamental group of a wedge of spaces is the free product of their fundamental groups. For the first theorem, the generalization to wedges then follows by the uniqueness of the rightmost and leftmost paths in each surface they span. For the second theorem, we obtain the generalization to wedges by gluing the disk diagrams in each surface together.

The following lemma is the crux of the analysis. We state it for rightmost paths but the symmetric version with leftmost paths holds with the same proof.

\begin{lemma}\label{L:removingstaircases}
At the end of step 2, the paths represented by $A_j^R$ and $A_k^R$ are rightmost. Furthermore, denoting by $\tau$ the turn between $A_j^R$ and $A_k^R$ at the end of step 2, the path represented by $A_j^R \tau A_k^R$ is homotopic to the one before these two steps, contains no spurs, contains at most two brackets, and these brackets contain both at least one $2$.
\end{lemma}

\begin{proof}
Since $A_j^R$ and $A_k^R$ are obtained from subpaths of rightmost curves, they are rightmost. Recall that $A_i^R$ is the rightmost path homotopic to $A_i$. Since $A^R_i$ is homotopic to $A_j^R\tau A_k^R$, the three paths represented by $A_j^R$, $A_k^R$ and $\overline{A_i^R}$ bound together a disk diagram. Since $A_i^R$ and $A^R_j$ are rightmost, by uniqueness of rightmost representatives of paths, they share a unique maximal subpath which is a common prefix, likewise for $A_i^R$ and $A^R_k$ where it is a common suffix. We consider the combinatorial triangle $\Delta$ obtained after removing these subpaths, see Figure~\ref{F:triangles}. We denote its endpoints by $a, b$ and $c$, where $a$ is the vertex between $A^R_j$ and $A^R_k$, $b$ the vertex between $A^R_k$ and $A^R_i$, and $c$ the last one. Let $x$ denote the first vertex of degeneracy of $\Delta$ on the path from $c'$ to $a$, and $y$ denote the vertex at the end of the first common path of $A^R_j$ and $A^R_k$ between $x$ and $a$. Then we have a pinched disk diagram between $y$ and $a$, the boundary of which is made of a subpath $\gamma_j$ of $A_j^R$ and a subpath $\gamma_k$ of $A_k^R$. Theorem~\ref{T:geodesics} tells us that this disk diagram is an alternating sequence of staircases and paths, which are connected by their tips. By definition, $y$ is the tip of a staircase. We claim that $y$ also lies on $A_k^L$. Otherwise, denoting by $\lambda_1$ and $\lambda_2$ the leftmost paths between $a$ and $y$ and between $y$ and $c$, the turn between $\lambda_1$ and $\lambda_2$ at $y$ would be $\overline{1}$, yielding a bracket with the other side of the staircase, a contradiction.

Since $\gamma_j$ and $\gamma_k$ are both rightmost, they have the same length and $\lambda_1$, the leftmost path homotopic to $\gamma_k$ is exactly $\overline{\gamma_j}$.  Therefore, during Step 1., the whole subdisk between $\gamma_j$ and $\gamma_k$ got removed.

\begin{figure}[t]
    \centering
    \def\svgwidth{\textwidth}
    \begingroup%
  \makeatletter%
  \providecommand\color[2][]{%
    \errmessage{(Inkscape) Color is used for the text in Inkscape, but the package 'color.sty' is not loaded}%
    \renewcommand\color[2][]{}%
  }%
  \providecommand\transparent[1]{%
    \errmessage{(Inkscape) Transparency is used (non-zero) for the text in Inkscape, but the package 'transparent.sty' is not loaded}%
    \renewcommand\transparent[1]{}%
  }%
  \providecommand\rotatebox[2]{#2}%
  \ifx\svgwidth\undefined%
    \setlength{\unitlength}{3723.87942619bp}%
    \ifx\svgscale\undefined%
      \relax%
    \else%
      \setlength{\unitlength}{\unitlength * \real{\svgscale}}%
    \fi%
  \else%
    \setlength{\unitlength}{\svgwidth}%
  \fi%
  \global\let\svgwidth\undefined%
  \global\let\svgscale\undefined%
  \makeatother%
  \begin{picture}(1,0.61726836)%
    \put(0,0){\includegraphics[width=\unitlength,page=1]{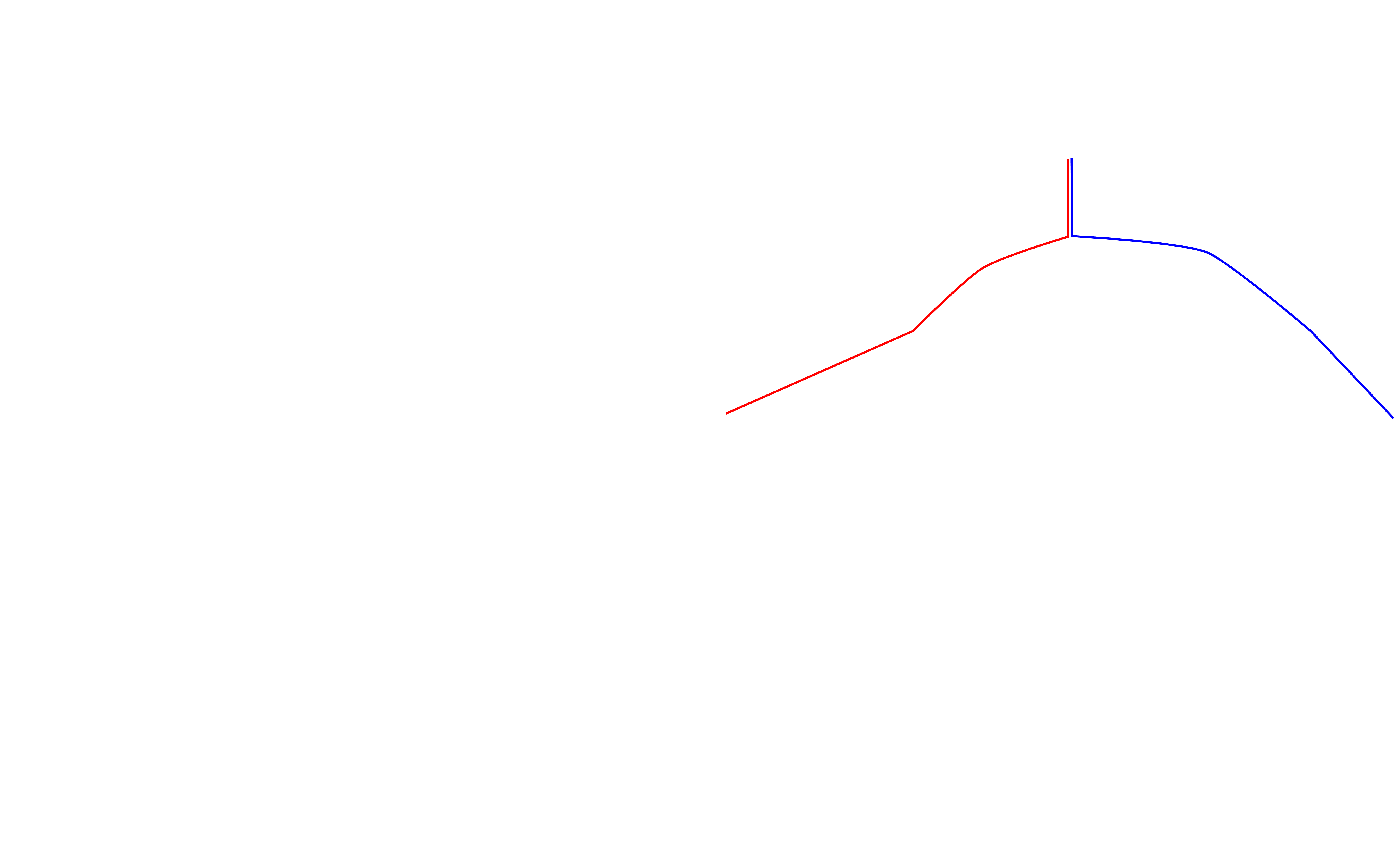}}%
    \put(0.64487751,0.4299007){\color[rgb]{1,0,0}\makebox(0,0)[lb]{\smash{$A_j^R$}}}%
    \put(0,0){\includegraphics[width=\unitlength,page=2]{triangles.pdf}}%
    \put(0.85626071,0.44828836){\color[rgb]{0,0,1}\makebox(0,0)[lb]{\smash{$A_k^R$}}}%
    \put(0,0){\includegraphics[width=\unitlength,page=3]{triangles.pdf}}%
    \put(0.10650491,0.10597447){\color[rgb]{0,0,1}\makebox(0,0)[lb]{\smash{$A_k^R$}}}%
    \put(0.25596499,0.04428766){\color[rgb]{1,0,1}\makebox(0,0)[lb]{\smash{$A_i^R$}}}%
    \put(0.32931555,0.29285774){\color[rgb]{0,0,0}\makebox(0,0)[lb]{\smash{$a$}}}%
    \put(0.10244835,0.07614481){\color[rgb]{0,0,0}\makebox(0,0)[lb]{\smash{$c'$}}}%
    \put(0.22410241,0.1856544){\color[rgb]{0,0,0}\makebox(0,0)[lb]{\smash{$x$}}}%
    \put(0,0){\includegraphics[width=\unitlength,page=4]{triangles.pdf}}%
    \put(0.61810656,0.11426077){\color[rgb]{0,0,1}\makebox(0,0)[lb]{\smash{$A_k^R$}}}%
    \put(0.77094255,0.04428766){\color[rgb]{1,0,1}\makebox(0,0)[lb]{\smash{$A_i^R$}}}%
    \put(0.61323461,0.07845107){\color[rgb]{0,0,0}\makebox(0,0)[lb]{\smash{$c'$}}}%
    \put(0.94506317,0.07999989){\color[rgb]{0,0,0}\makebox(0,0)[lb]{\smash{$b'$}}}%
    \put(0,0){\includegraphics[width=\unitlength,page=5]{triangles.pdf}}%
    \put(0.49657429,0.16827509){\color[rgb]{0,0,0}\makebox(0,0)[lb]{\smash{Step 1}}}%
    \put(0.77401538,0.20648744){\color[rgb]{0,0,0}\makebox(0,0)[lb]{\smash{$a'$}}}%
    \put(0,0){\includegraphics[width=\unitlength,page=6]{triangles.pdf}}%
    \put(0.3484125,0.13834605){\color[rgb]{1,0,0}\makebox(0,0)[lb]{\smash{$A_j^R$}}}%
    \put(0.87357863,0.13834605){\color[rgb]{1,0,0}\makebox(0,0)[lb]{\smash{$A_j^R$}}}%
    \put(0,0){\includegraphics[width=\unitlength,page=7]{triangles.pdf}}%
    \put(0.25596499,0.35621848){\color[rgb]{1,0,1}\makebox(0,0)[lb]{\smash{$A_i^R$}}}%
    \put(0.32931555,0.60478856){\color[rgb]{0,0,0}\makebox(0,0)[lb]{\smash{$a$}}}%
    \put(0.43404085,0.38373541){\color[rgb]{0,0,0}\makebox(0,0)[lb]{\smash{$c'$}}}%
    \put(0.22410241,0.44810675){\color[rgb]{0,0,0}\makebox(0,0)[lb]{\smash{$x$}}}%
    \put(0,0){\includegraphics[width=\unitlength,page=8]{triangles.pdf}}%
    \put(0.77094255,0.35621848){\color[rgb]{1,0,1}\makebox(0,0)[lb]{\smash{$A_i^R$}}}%
    \put(0.95437558,0.38777776){\color[rgb]{0,0,0}\makebox(0,0)[lb]{\smash{$c'$}}}%
    \put(0.61867891,0.38672246){\color[rgb]{0,0,0}\makebox(0,0)[lb]{\smash{$b'$}}}%
    \put(0,0){\includegraphics[width=\unitlength,page=9]{triangles.pdf}}%
    \put(0.49657429,0.48020591){\color[rgb]{0,0,0}\makebox(0,0)[lb]{\smash{Step 1}}}%
    \put(0.2618699,0.50626564){\color[rgb]{0,0,0}\makebox(0,0)[lb]{\smash{$a'$}}}%
    \put(0,0){\includegraphics[width=\unitlength,page=10]{triangles.pdf}}%
    \put(0.1856544,0.5231925){\color[rgb]{0,0,0}\makebox(0,0)[lb]{\smash{$y$}}}%
    \put(0,0){\includegraphics[width=\unitlength,page=11]{triangles.pdf}}%
    \put(0.13051376,0.4299007){\color[rgb]{1,0,0}\makebox(0,0)[lb]{\smash{$A_j^R$}}}%
    \put(0,0){\includegraphics[width=\unitlength,page=12]{triangles.pdf}}%
    \put(0.71686805,0.45748825){\color[rgb]{0,0,0}\makebox(0,0)[lb]{\smash{$x$}}}%
    \put(0,0){\includegraphics[width=\unitlength,page=13]{triangles.pdf}}%
    \put(0.72912079,0.49705864){\color[rgb]{0,0,0}\makebox(0,0)[lb]{\smash{$a'$}}}%
    \put(0.77946128,0.45836595){\color[rgb]{0,0,0}\makebox(0,0)[lb]{\smash{$\tau \geq 2$}}}%
    \put(0.34189693,0.44828836){\color[rgb]{0,0,1}\makebox(0,0)[lb]{\smash{$A_k^R$}}}%
    \put(0,0){\includegraphics[width=\unitlength,page=14]{triangles.pdf}}%
    \put(0.09090864,0.38585441){\color[rgb]{0,0,0}\makebox(0,0)[lb]{\smash{$b'$}}}%
    \put(0.43812592,0.07305554){\color[rgb]{0,0,0}\makebox(0,0)[lb]{\smash{$b'$}}}%
    \put(0.05935414,0.31716294){\color[rgb]{0,0,0}\makebox(0,0)[lb]{\smash{$b$}}}%
    \put(0.43434877,0.318031){\color[rgb]{0,0,0}\makebox(0,0)[lb]{\smash{$c$}}}%
    \put(0.58104807,0.31803099){\color[rgb]{0,0,0}\makebox(0,0)[lb]{\smash{$b$}}}%
    \put(0.96125098,0.31629491){\color[rgb]{0,0,0}\makebox(0,0)[lb]{\smash{$c$}}}%
    \put(0.46907051,0.02984065){\color[rgb]{0,0,0}\makebox(0,0)[lb]{\smash{$b$}}}%
    \put(0.05501391,0.00813958){\color[rgb]{0,0,0}\makebox(0,0)[lb]{\smash{$c$}}}%
    \put(0.56715939,0.00640349){\color[rgb]{0,0,0}\makebox(0,0)[lb]{\smash{$b$}}}%
    \put(0.95604272,0.01074371){\color[rgb]{0,0,0}\makebox(0,0)[lb]{\smash{$c$}}}%
    \put(0.29285774,0.53938201){\color[rgb]{0,0,0}\makebox(0,0)[lb]{\smash{$\gamma_j$}}}%
    \put(0.19563689,0.58886045){\color[rgb]{0,0,0}\makebox(0,0)[lb]{\smash{$\gamma_k$}}}%
    \put(0.2919897,0.22862254){\color[rgb]{0,0,0}\makebox(0,0)[lb]{\smash{$\gamma_j$}}}%
    \put(0.18869254,0.26768448){\color[rgb]{0,0,0}\makebox(0,0)[lb]{\smash{$\gamma_k$}}}%
    \put(0.18391831,0.21069695){\color[rgb]{0,0,0}\makebox(0,0)[lb]{\smash{$y$}}}%
  \end{picture}%
\endgroup%

    \caption{The two non-degenerate cases for the combinatorial triangles in the proof of Lemma~\ref{L:removingstaircases}. Staircases are colored. Step $2$ will remove the spur or the ladder.}
    \label{F:triangles}
\end{figure}

We denote by $\Delta'$ the disk diagram obtained after doing Step 1, with endpoints $a'$, $b'$ and $c'$ as before, and we distinguish cases depending on its orientation. If $\Delta'$ is degenerate, i.e., has no quads in its interior, then after removing possibly one sequence of spurs in step 2., we have $A_j^R\tau A_k^R=A_i^R$ and thus we have successfully computed a rightmost form for $A_i$.

If the boundary of $\Delta'$, oriented $c' \rightarrow a' \rightarrow b' \rightarrow c'$, is oriented clockwise, then $a'$ is on the path between $x$ and $y$, since $A^L_k$ and $A^R_j$ can not coincide below $x$, see Figure~\ref{F:triangles}, top. After removing the sequence of spurs between $a'$ and $x$ in step 2., the turn between $A_j^R$ and $A_k^R$ is at least $2$, since otherwise there would be a spur in $A_j^R$ or $A_k^R$.

If the boundary of $\Delta'$, oriented $c' \rightarrow a' \rightarrow b' \rightarrow c'$, is oriented counter-clockwise, we observe that $x$ also lies on $A^L_k$. Indeed, otherwise, since $A^L_k$ is at distance at most one from $A^R_k$ and $\Delta'$ is not degenerate, there would be a common vertex $z$ that would be both on $A_j^R$ and $A_k^L$ and such that the paths from $z$ to $a$ are different, since one contains $x$ and not the other. This would contradict uniqueness of canonical paths. Therefore, the entire pinched disk between $x$ and $a$ has been removed in Step 1. Now we look at the maximal (possibly empty) staircase having $a'$ as a tip and being bounded by pair of subpaths $\alpha_j$ and $\alpha_k$ of $A_j^R$ and $A_k^R$, except for one edge. If there is no such staircase, nothing happens in step 2, and we are done. If there is such a staircase, we claim that most of it got already removed when we get to Step 2. If $\alpha_k$ contains one or more $1$s, then the start of the leftmost form of $\alpha_k$ coincides with the start of $\alpha_j^{-1}$, and thus this subword containing this $1$ has also been eliminated during step 1, see Figure~\ref{F:reductions}. Note that the slight offset of the indices in Step 1 is due to the difference in length between $\alpha_j$ and $\alpha_k$, and that the new turn $\tau$ that we introduce is the missing edge to close the staircase.  Therefore, $\alpha_k$ contains no $1$, and the last turn of $\alpha_j$ is a $1$. Therefore the entire staircase is a ladder, as pictured in the left of Figure~\ref{F:snakes} and on the bottom right of Figure~\ref{F:triangles}, which gets removed during step 2.

Let $\tau$ denote the turn between $A_j^R$ and $A_k^R$ after steps 1 and 2. Then in the clockwise case, $\tau$ is at least $2$ and thus at most one bracket can exist in $A_j^R \tau A_k^R$. In the counter-clockwise case, the only brackets that can exist in $A_j^R \tau A_k^R$ must contain a $2$, as otherwise the staircase of the previous paragraph would not have been maximal. There are at most two such brackets: one finishing at $\tau$ and one starting at $\tau$, since otherwise such a bracket would have been present in $A_j^R$ or $A_k^R$. As the reductions correspond to flattening brackets and removing spurs on a disk diagram, they are homotopies. Finally, since the degenerate parts of the disk diagram have been removed, there are no spurs.
\end{proof}

We now prove that the reduction algorithm does indeed compute rightmost and leftmost forms.

\begin{lemma}\label{L:reduction}
The \textsc{Reduction Algorithm} runs in polynomial time. The straight-line program that it outputs has the property that every pair of characters $A_i^R$ and $A_i^L$ encode a pair of turn sequence corresponding respectively to the rightmost and leftmost paths homotopic to $A_i$. In particular, at the top level, $\mathbb{A}^R$ and $\mathbb{A}^L$ are the rightmost and leftmost paths homotopic to the compressed input walk.
\end{lemma}

\begin{proof}
By Lemma~\ref{L:removingstaircases}, when the algorithm reaches step 3, the path $A_j^R A_k^R$ may not be reduced for two reasons: either the turn between $A_j^R$ and $A_k^R$ is a $\overline{1}$, or there is one or two brackets with at least one $2$ between them. Note that the two cases are exclusive. We argue that no spurs or brackets remain after the two reductions a and b. For the first one (reduction a)), note that $x,y \neq \overline{2}$ by definition and $x$ and $y$ are neither $0$ nor $\overline{1}$ because $A_j^R$ is rightmost. We may have $x=1$ or $y=1$, which may create at most two brackets, but note that these brackets are disjoint and contain each at least one $2$.

When we reach step b, we therefore have at most two brackets, each containing at least one $2$. We conclude by proving that no new brackets nor spurs are created during an application of this step b. Note that $x$ and $y$ cannot be $\overline{1}$ because the curves are rightmost, cannot be $0$ because they contain no spurs, and cannot be $1$ because otherwise there would be a bracket without a $2$. They could be $3$, but the $\overline{2}$ blocks the resulting $2$ from participating in a bracket since there is no torus surface among the wedge summands. Therefore, the only way for $x-1$ or $y-1$ to be involved in a bracket is if one of them is $1$, in which case it bounds a second bracket that already existed before the reduction. After reducing that one, there are no more possible spurs nor brackets. In particular, step 3.b. only loops a constant number of times, and thus the whole inductive step is polynomial. In turn, computing the full straight-line programs for both reduced words $\mathbb{A}^L$ and $\mathbb{A}^R$ is also polynomial.
\end{proof}

We are now ready to prove Theorem~\ref{T:triviality}.

\begin{proof}[Proof of Theorem~\ref{T:triviality}]
Using Lemma~\ref{L:quadsystem}, we first turn our walk $w$ into a walk on a quad system. Then, using Lemma~\ref{L:reduction}, we compute a compressed rightmost loop $w'$ homotopic to $w$ in this quad system. By Theorem~\ref{T:uniqueness}, there is a unique such rightmost loop. Combined with the data of the first and last edge, this is our canonical form. If $w$ is homotopic to a trivial walk, it is homotopic to the empty loop, which is rightmost. Therefore, to test contractibility, it suffices to test whether $w'$ is the empty word.
\end{proof}

Finally, let us comment on how to adapt this algorithm in the case where there are tori in the wedge. In tori summands, we dispense from using turn sequences and simply reduce to a one-vertex one-face graph, which therefore has two edges $a$ and $b$. Homotopy in tori is a simple matter since the fundamental group is $\mathbb{Z}^2$, corresponding exactly to how many times the oriented edges $a$ and $b$ are taken. Therefore, we directly use for our encoding a pair $(k,\ell)$ representing this, with the convention that $\overline{(k,\ell)}=(-k,-\ell)$. The reduction algorithm has an additional rule stipulating that $(k_1,\ell_1)(k_2,\ell_2)$ should be reduced to $(k_1+k_2,\ell_1+\ell_2)$. Any reduction to $(0,0)$ triggers a search for a deeper sequences of spurs and staircases as in Steps 1 to 3.

\section*{Proofs of the main theorems}

 We now have all the tools to prove our main theorems. Theorem~\ref{t:2mccontract} is a subcase of Theorem~\ref{T:triviality} when the the wedge of surfaces consists of a single surface.

\begin{proof}[Proof of Theorem~\ref{t:nsm}]
We can make $\Delta$ reduced in polynomial time~\cite{Erickson2013}. By Proposition~\ref{P:retriangulation}, one can compute in polynomial time a new triangulation of the surface $S$ so that the multi-curve $\Delta$ has polynomial complexity and the curve $\gamma$ is encoded as a straight-line program. By Proposition~\ref{P:cappingoff}, we can compute a wedge of surfaces and circles $K$ and a walk $w$ on $K$ so that $w$ is contractible in $K$ if and only if $\gamma$ belongs to the normal subgroup generated by the curves $\Delta$ in $S$. The simplicial complex $K$ can be transformed to a wedge of combinatorial surfaces and circles easily. Then the theorem follows from Theorem~\ref{T:triviality}.
\end{proof}

\begin{proof}[Proof of Theorem~\ref{t:3mcontract}]
By Proposition~\ref{P:orientability}, we can reduce our problem to the case of orientable manifolds. Note that the size of the manifold at most doubles in this process. Then, Proposition~\ref{P:from3dto2d} reduces the problem to the disjoint normal subgroup membership problem. The \textbf{NP} certificate consists of the certificate in that proposition, and the corresponding algorithm is to guess this certificate. The disjoint normal subgroup membership problem is then solved using Theorem~\ref{t:nsm}.
\end{proof}

\section*{Acknowledgements}
The authors would like to thank Mark Bell, Ben Burton, and Jeff Erickson for  helpful discussions. We also thank an anonymous referee of~\cite{cdvparsaJ} for suggesting the use of a maximal compression body as a simple exponential-time algorithm for deciding contractability of arbitrary (non-compressed) curves on the boundary of a 3-manifold.

\end{document}